\documentclass[conference]{IEEEtran}

\IEEEoverridecommandlockouts

\usepackage[
    n, 
    advantage,
    operators,
    sets,
    adversary,
    landau,
    probability,
    notions,
    logic,
    ff,
    mm,
    primitives,
    events,
    complexity,
    oracles,
    asymptotics,
    keys
]{cryptocode}


\usepackage{xspace}
\newcommand{\kemuauth}{$\mathsf{KEMUAuth}$\xspace}

\usepackage{enumitem}
\setlist[enumerate]{label=\arabic*.}

\usepackage{multirow}
\usepackage{cite}
\usepackage{amsmath,amssymb,amsfonts}
\usepackage{xurl}
\usepackage[hidelinks]{hyperref}
\usepackage{amsthm}

\theoremstyle{plain}
\newtheorem{theorem}{Theorem}[section]
\newtheorem{lemma}{Lemma}[section]
\newtheorem{definition}{Definition}[section]

\usepackage{algorithmic}
\usepackage{graphicx}
\usepackage{textcomp}
\usepackage{url}
\usepackage{xcolor}

\def\BibTeX{{\rm B\kern-.05em{\sc i\kern-.025em b}\kern-.08em
    T\kern-.1667em\lower.7ex\hbox{E}\kern-.125emX}}

\begin{document}

\title{A Drop-in KEM Replacement for Client Signatures in Post-Quantum SSH%
\thanks{%
Extended version of the paper accepted at IEEE ICNP 2026.
\textcopyright~2026 IEEE. Personal use of this material is permitted.
Permission from IEEE must be obtained for all other uses, in any current or
future media, including reprinting/republishing this material for advertising
or promotional purposes, creating new collective works, for resale or
redistribution to servers or lists, or reuse of any copyrighted component of
this work in other works.%
}}

\author{
\IEEEauthorblockN{
Hongbo Liu\IEEEauthorrefmark{1}\IEEEauthorrefmark{2},
Yufan Su\IEEEauthorrefmark{3},
Jiangxia Ge\IEEEauthorrefmark{4},
Qionglu Zhang\IEEEauthorrefmark{1}\IEEEauthorrefmark{2},
Zhaoxuan Li\IEEEauthorrefmark{1}\IEEEauthorrefmark{2},
\\
Xianhui Lu\IEEEauthorrefmark{1}\IEEEauthorrefmark{2},
Li Song\IEEEauthorrefmark{1}\IEEEauthorrefmark{2},
Wenhua Gao\IEEEauthorrefmark{5},
Li Zhou\IEEEauthorrefmark{6}
}

\IEEEauthorblockA{
\IEEEauthorrefmark{1}
State Key Laboratory of Cyberspace Security Defense,
Institute of Information Engineering, CAS, Beijing, China
}

\IEEEauthorblockA{
\IEEEauthorrefmark{2}
School of Cyber Security,
University of Chinese Academy of Sciences, Beijing, China
}

\IEEEauthorblockA{
\IEEEauthorrefmark{3}
School of Control and Computer Engineering,
North China Electric Power University, Beijing, China
}

\IEEEauthorblockA{
\IEEEauthorrefmark{4}
China Telecom Quantum Information Technology Group Co., Ltd.,
Hefei, China
}

\IEEEauthorblockA{
\IEEEauthorrefmark{5}
Beijing Certificate Authority Co., Ltd. (BJCA),
Beijing, China
}

\IEEEauthorblockA{
\IEEEauthorrefmark{6}
Institute of Software,
Chinese Academy of Sciences, Beijing, China
}

\IEEEauthorblockA{
Corresponding authors: Qionglu Zhang and Zhaoxuan Li
(\{zhangqionglu,lizhaoxuan\}@iie.ac.cn)
}
}

\maketitle

\begin{abstract}
The transition to post-quantum cryptography is reshaping the Secure Shell (SSH) protocol for remote administration. Post-quantum key exchange has been deployed in OpenSSH and is being standardized, while SSH authentication largely remains a signature-replacement effort. This path preserves the familiar public-key credential model, but inherits the size and computation overhead of post-quantum signatures, which can increase latency, traffic, and server-side load. KEM-based authentication offers a natural alternative to this signature-centric path, and SSH makes this especially attractive at the user-authentication layer, which is method-extensible, separated from transport-layer key exchange and host-key authentication, and already protected by the established channel.

We present a drop-in KEM-based user-authentication method for SSH that replaces client public-key signatures with a session-bound challenge--response proof. The method fits into SSH's existing user-authentication framework, preserving the public-key credential model and enabling incremental deployment alongside existing methods. We provide a reduction-based security argument in the post-quantum ACCE framework, implement the design in OpenSSH using liboqs, and evaluate it under representative RTTs, TCP initial-window settings, and post-quantum migration configurations. Our results show that KEM-based authentication is competitive with compact signature-based authentication under representative network settings, while reducing median handshake latency by up to about 10\% against large-signature hybrid baselines. The advantages are clearer when post-quantum signatures stress transmission or computation: median latency under small TCP initial windows falls by up to 7.3\% versus ML-DSA and 17.9\% versus SLH-DSA, while server-side online cryptographic cost is 59.1\% lower than that for ML-DSA in the same NIST category.
\end{abstract}


\bstctlcite{IEEEexample:BSTcontrol}

\section{Introduction}
\label{sec:intro}

The Secure Shell (SSH) protocol is widely deployed for secure remote administration. Its security relies on public-key mechanisms to establish the encrypted transport channel, authenticate the server, and commonly authenticate users through public-key credentials inside that channel~\cite{rfc4253,acce,pqacce}. As quantum computing threatens the assumptions underlying these mechanisms, SSH deployments and standards are beginning to migrate toward post-quantum cryptography~\cite{ms_pqssh}.

This transition has so far been asymmetric. On the key-exchange side, post-quantum KEMs fit the role of establishing fresh shared secrets, and SSH has already begun to adopt them. OpenSSH has deployed transport-layer methods such as \texttt{mlkem768x25519-sha256}~\cite{openssh}, and related SSH key-exchange mechanisms are moving through standardization~\cite{draft_ietf_sshm_mlkem_hybrid,draft_harrison_sshm_mlkem}. Authentication has followed a more conservative path. Existing prototypes and standardization efforts largely preserve the \texttt{publickey} method by replacing classical signatures, such as Ed25519~\cite{rfc8709} and RSA~\cite{rfc8332}, with post-quantum signatures. This path is easy to integrate, but it also makes authentication inherit their size and verification profile. Prior measurements show that signature-based post-quantum SSH handshakes can increase latency by up to 50\% in the worst case under realistic network conditions, especially when large authentication objects exceed the TCP initial congestion window~\cite{assesstlsssh,study2026}. Password authentication avoids client-side signatures, but gives up the registered public-key credential model and its public-key proof-of-possession property. SSH therefore lacks a non-signature option with these properties~\cite{assesstlsssh}.

KEMs provide a natural basis for authentication because encapsulation to a registered public key allows only the corresponding secret-key holder to derive the secret needed for a session-bound response. Signature-free authentication has long appeared in AKE and secure-channel designs~\cite{krawczyk1996skeme,bellare1998modular,perrin2016doubleratchet,perrin2018noise,donenfeld2017wireguard}, and recent post-quantum work has turned it into practical protocol designs~\cite{hulsing2021pqwireguard,kemtls}. These designs show that KEM-based authentication can reduce bandwidth and server CPU costs in post-quantum settings, as illustrated by KEMTLS and the AuthKEM IETF draft~\cite{kemtls,authkem}. However, they typically integrate KEM credentials into the handshake or key schedule, so deployment requires changes to core protocol messages, transcript binding, key derivation, and interoperability mechanisms in existing implementations.

SSH makes this picture less rigid. Its architecture separates user authentication from the transport handshake that establishes the encrypted channel and authenticates the server. By the time user authentication begins, the channel keys and session identifier have already been established, and the authentication exchange runs inside the protected channel~\cite{rfc4251,rfc4252}. At the same time, this layer is method-extensible and policy-driven~\cite{rfc4252}, and SSH already accommodates challenge--response mechanisms such as \texttt{keyboard-interactive}~\cite{rfc4256}. This combination of separation and extensibility suggests a more local path for bringing KEM-based authentication into SSH within its existing authentication framework. 

\subsection{Our Contributions}
\label{subsec:contributions}

 In this paper, we present a drop-in user-authentication method that adapts KEM-based authentication to post-quantum SSH. We study  this adaptation end to end through a post-quantum ACCE analysis, an OpenSSH implementation, and controlled network evaluation. Our contributions are summarized as follows:

\textbf{Drop-in KEM authentication without extra SSH round trips.}
 We design \kemuauth to preserve SSH's public-key credential model while replacing client signatures with a deniable, session-bound proof of KEM secret-key possession. The proof is bound to the SSH session identifier, the challenge ciphertext, and the authentication context, giving it the same session-binding role as the signed request in the existing \texttt{publickey} method. \kemuauth follows the commonly used two-step \texttt{publickey} flow with \texttt{USERAUTH\_PK\_OK} without adding SSH round trips. At the user-authentication layer, \kemuauth allows servers to authorize KEM credentials alongside existing signature keys and to combine them with SSH's existing multi-method authentication policies.

\textbf{Post-quantum ACCE security for KEMUAuth.}
We analyze SSH with \kemuauth in the post-quantum ACCE framework, building on prior ACCE analyses of SSH~\cite{pqacce}. The main new proof obligation is server acceptance of a client authenticated through \kemuauth, which extends the existing server-authenticated SSH analysis to mutual authentication across the transport and user-authentication layers. We reduce this case to the IND-CCA security of the client's long-term KEM credential and the PRF security of the session-bound response function. Since \kemuauth runs inside the established encrypted channel and does not affect transport-layer key derivation, the authentication argument composes with the existing SSH channel-security proof.

\textbf{Practical implementation and controlled migration evaluation.}
We use liboqs~\cite{liboqs} to implement \kemuauth as a new OpenSSH user-authentication method that coexists with signature credentials and supports standalone KEM, hybrid, and SSH-native multi-step migration configurations. We evaluate it under controlled network emulation across representative RTTs, TCP initial-window settings, and practical post-quantum authentication choices, measuring handshake latency, transmission effects, and server-side online authentication cost. The results show that \kemuauth closely matches Ed25519 and ML-DSA in end-to-end latency under common network settings, with clearer benefits emerging when post-quantum signatures impose larger transmission or verification costs. In particular, under small TCP initial windows, it reduces median handshake latency by 7.0--7.3\% relative to the ML-DSA baseline and lowers the measured server-side online cryptographic cost by 59.1\% compared with ML-DSA in the same NIST security category.

\subsection{Related Work}
\label{subsec:related}

Post-quantum SSH migration has produced a clear division between key exchange and authentication. For key exchange, KEM-based mechanisms have moved from prototypes into deployment and standardization. Mainline OpenSSH treats post-quantum key agreement as a baseline capability, and large-scale services such as GitHub have enabled hybrid post-quantum SSH key exchange for production access~\cite{openssh_pq,github_pqssh}. Standardization has followed this direction, with RFC~9941 documenting the widely deployed \texttt{sntrup761x25519-sha512} hybrid KEX method and ongoing Internet-Drafts specifying ML-KEM-based SSH key-exchange variants~\cite{rfc9941,draft_ietf_sshm_mlkem_hybrid,draft_harrison_sshm_mlkem}. For authentication, existing proposals mainly preserve SSH's signature-based public-key model by adding post-quantum signature formats, including ML-DSA public keys, composite ML-DSA signatures, and SSH profiles that mandate ML-DSA-based authentication~\cite{draft_sfluhrer_ssh_mldsa,draft_rpe_ssh_mldsa,draft_becker_cnsa2_ssh_profile}. Research prototypes and measurements mirror this division, evaluating post-quantum KEX and signature-based authentication in OQS-OpenSSH and liboqs-based prototypes, and quantifying their latency, bandwidth, and deployment impact~\cite{oqsopenssh,crockett2019prototyping,assesstlsssh,zavacke2025practical}.

Formal analyses of post-quantum SSH have primarily treated the transport layer as the security-critical boundary. Tran et~al. give a symbolic AKE-style analysis of hybrid SSH key exchange~\cite{tran2024formal}, while Ben\v{c}ina et~al. provide a post-quantum ACCE analysis of SSH that captures \textit{harvest-now, decrypt-later} attacks~\cite{pqacce}. Other formal verification work similarly focuses on hybrid key exchange and the confidentiality or integrity of the resulting channel~\cite{blanchet2024post}. These analyses provide the foundation for reasoning about post-quantum SSH transport security, but their authentication guarantees are centered on the transport-layer exchange and server authentication, rather than on proof of possession in the user-authentication layer. A related but orthogonal line is the privacy-preserving SSH authentication scheme of Roy et~al., which uses anonymous multi-KEM and private set intersection to hide authorized-key membership~\cite{roy2022ssh}. Its goal is authorized-key privacy, rather than a general post-quantum KEM credential for standard SSH user authentication.

KEM-based authentication has been explored beyond SSH in AKE and secure-channel protocols. In TLS-style designs, KEMTLS~\cite{kemtls} and AuthKEM~\cite{authkem} replace handshake signatures with certified long-term KEM public keys, while post-quantum WireGuard~\cite{hulsing2021pqwireguard} redesigns its handshake to support KEM-based authentication. Similar directions are being considered for IKEv2/IPsec and EDHOC, including KEM-based, signature-free, and mixed signature/KEM authentication mechanisms for constrained or migration-oriented settings~\cite{draft_wang_ipsecme_kem_auth_ikev2,draft_pocero_authkem_edhoc,draft_pocero_authkem_ikr_edhoc}. At the AKE-construction level, protocols such as PQuAKE~\cite{draft_uri_lake_pquake}, Muckle\#~\cite{battarbee2025mucklesharp}, and HetAKE~\cite{heterogeneousake} study implicit, hybrid, or heterogeneous authentication mechanisms for post-quantum migration. Together, these works establish KEM-based authentication as a viable design paradigm, while also illustrating that existing approaches typically require redesigning the handshake or key schedule to place KEM credentials inside the core protocol flow.

Overall, existing post-quantum SSH work has treated KEMs primarily as transport-layer key-exchange mechanisms, while existing KEM-authentication work has largely placed KEM credentials inside AKE handshakes or key schedules. This leaves SSH's user-authentication layer as a distinct and comparatively underexplored point in the design space for post-quantum KEM-based authentication.

\section{Preliminaries}
\label{sec:preliminary}

\subsection{Notation}
\label{subsec:nota}

The security parameter is denoted by $\lambda\in\mathbb{N}$. For a finite set $S$, $x \xleftarrow{\$} S$ denotes sampling $x$ uniformly at random from $S$, and $y \leftarrow \mathsf{Alg}(x)$ denotes assignment of the output of algorithm $\mathsf{Alg}$ on input $x$. Unless stated otherwise, probabilities are taken over the randomness of the experiment, the algorithms, and the adversary. For bit strings $x$ and $y$, $x\parallel y$ denotes concatenation, and $\{0,1\}^{\ell}$ denotes the set of $\ell$-bit strings. The symbol $\bot$ denotes failure or rejection. PPT and QPT stand for probabilistic polynomial-time and quantum polynomial-time, respectively. We write $\mathsf{negl}(\lambda)$ for an unspecified negligible function, and $\mathsf{Adv}^{\mathsf{notion}}_{\Pi}(\mathcal{A})$ for the advantage of adversary $\mathcal{A}$ against primitive or protocol $\Pi$ in security notion $\mathsf{notion}$.

\subsection{KEMs}
\label{subsec:kem}

A key encapsulation mechanism (KEM) is a public-key primitive that allows a party to establish a shared secret with the holder of a corresponding secret key. Formally, a KEM consists of three probabilistic polynomial-time algorithms
\[
    \mathsf{KEM} = (\mathsf{KeyGen}, \mathsf{Encap}, \mathsf{Decap})
\]
defined as follows:
\begin{itemize}
    \item $\mathsf{KeyGen}(1^\lambda)$ outputs a public/secret key pair $(pk,sk)$.
    \item $\mathsf{Encap}(pk)$ outputs a ciphertext $ct$ and a shared secret $K$.
    \item $\mathsf{Decap}(sk,ct)$ outputs either a shared secret $K'$ or the failure symbol $\bot$.
\end{itemize}
Correctness requires that, except with negligible probability, if $(pk,sk)\leftarrow \mathsf{KeyGen}(1^\lambda)$ and $(K,ct)\leftarrow \mathsf{Encap}(pk)$, then $\mathsf{Decap}(sk,ct)=K$.

The standard security requirement is indistinguishability of the encapsulated key from random. For $\mathsf{atk}\in\{\mathsf{cpa},\mathsf{cca}\}$, let
$G^{\mathsf{ind}\mbox{-}\mathsf{atk}}_{\mathsf{KEM},\mathcal{A}}$
denote the usual KEM indistinguishability experiment for adversary $\mathcal{A}$, where $\mathcal{A}$ receives a challenge ciphertext $ct^\star$ together with either the real encapsulated key or an independent random key. In the $\mathsf{IND}\mbox{-}\mathsf{CPA}$ experiment, $\mathcal{A}$ has no decapsulation access, while in the $\mathsf{IND}\mbox{-}\mathsf{CCA}$ experiment it may query a decapsulation oracle except on $ct^\star$. The advantage of $\mathcal{A}$ is
\[
  \mathsf{Adv}^{\mathsf{ind}\mbox{-}\mathsf{atk}}_
  {\mathsf{KEM}}(\mathcal{A})
  =
  \left|
  \Pr\!\left[
    G^{\mathsf{ind}\mbox{-}\mathsf{atk}}_
    {\mathsf{KEM},\mathcal{A}}(1^\lambda) = 1
  \right]
  - \frac{1}{2}
  \right|.
\]
A KEM is $\mathsf{ind}\mbox{-}\mathsf{atk}$ secure if this advantage is negligible for every efficient adversary $\mathcal{A}$.
\par
In our SSH analysis, the ephemeral KEM used in transport-layer key exchange requires $\mathsf{IND\mbox{-}CPA}$ security, while the long-term client KEM used for authentication requires the stronger $\mathsf{IND\mbox{-}CCA}$ security.
Formal experiments for these notions, together with the other cryptographic assumptions used in the analysis, are collected in Appendix~\ref{app:defs}.

\subsection{Baseline SSH Publickey Flow}
\label{subsec:pq-ssh-overview}

\begin{figure}[htbp]
\centering
\includegraphics[width=\columnwidth]{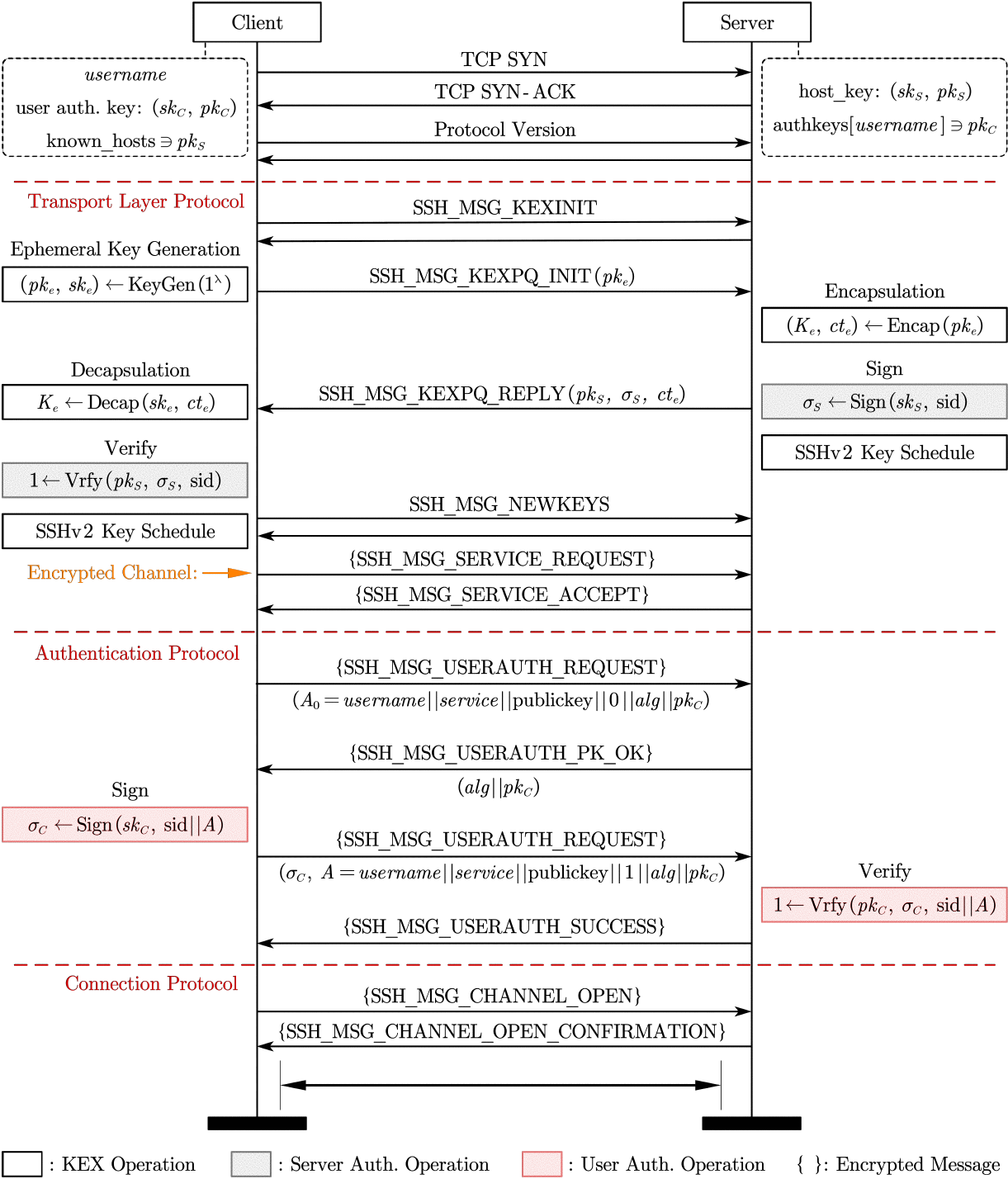}
\caption{Baseline SSH flow with post-quantum transport-layer key exchange and signature-based \texttt{publickey} user authentication. The flow follows RFC~4253~\cite{rfc4253} and RFC~4252~\cite{rfc4252}, abstracting recent OpenSSH post-quantum key-exchange suites~\cite{openssh_pq}. The key schedule is simplified: the first key exchange defines the session identifier $\mathsf{sid}$ as the exchange hash, which binds the version strings, \texttt{KEXINIT} messages, server host key, key-exchange values, and shared secret. The figure expands only the \texttt{publickey} authentication layer, where the client signs $\mathsf{sid}$ concatenated with the authentication request context $A$, excluding the signature field.}
\label{fig:pq-ssh-overview}
\end{figure}

Figure~\ref{fig:pq-ssh-overview} summarizes the SSH baseline considered in this work. SSH consists of the Transport Layer Protocol, the User Authentication Protocol, and the Connection Protocol. In post-quantum migration, the key-exchange and authentication primitives may be replaced by post-quantum or hybrid alternatives, while this layering remains unchanged. The transport layer establishes the session identifier $\mathsf{sid}$ and the encrypted channel. User authentication verifies the client inside that channel. Our design targets the signature proof in the \texttt{publickey} method, the dominant baseline for public-key user authentication.

During the transport phase, the client and server exchange protocol versions, negotiate algorithms via \texttt{SSH\_MSG\_KEXINIT}, and perform key exchange with server host-key authentication. After \texttt{SSH\_MSG\_NEWKEYS}, all subsequent service-request, user-authentication, and connection-layer messages are protected by the established encrypted channel. In the standard \textsf{publickey} method, the client typically first sends an unsigned \texttt{SSH\_MSG\_USERAUTH\_REQUEST} containing the username, requested service, method name, public-key algorithm, and public key. If the server considers the key acceptable, it replies with \texttt{SSH\_MSG\_USERAUTH\_PK\_OK}. The client then sends a signed \texttt{SSH\_MSG\_USERAUTH\_REQUEST}, where the signature $\sigma_C$ is computed over $\mathsf{sid}$ and the authentication request context $A$. SSH also permits a direct signed request without the preliminary acceptability query, but the two-step form is commonly used to avoid unnecessary signing or key-unlocking operations. Once user authentication succeeds, the connection protocol carries channel-opening and application data over the secure channel.

\subsection{ACCE Security Model}
\label{subsec:security-model}

We analyze our protocol in the authenticated and confidential channel establishment (ACCE) framework. An ACCE protocol first establishes peer identities, a session identifier, and channel keys during the handshake, and then uses the derived keys to protect application data with stateful authenticated encryption. This abstraction matches SSH: the transport layer establishes $\mathsf{sid}$ and channel keys, user authentication completes client authentication inside the protected channel, and the connection protocol carries application data. 
Our model follows the SSH ACCE formulation of~\cite{pqacce}, including its session variables, adversarial oracle interface, and matching-session structure, with the full experiment and our mutual-authentication extension specified in Appendix~\ref{app:acce}.

ACCE security consists of authentication and channel security. Authentication rules out malicious acceptance. A session $\pi_i^s$ accepts maliciously if it accepts believing that its peer is an uncorrupted party, but there is no unique matching session at that peer with the corresponding role, ciphersuite, and session identifier. In our SSH setting, this condition is required in both directions, meaning that the client must not accept without an honest matching server session and the server must not accept without an honest matching client session. For an adversary $\mathcal{A}$, we define
\[
  \mathsf{Adv}^{\mathsf{acce\mbox{-}mu\mbox{-}auth}}_{\mathsf{SSH}}(\mathcal{A})
  =
  \Pr\!\left[
    \exists\, \pi_i^s :
    \pi_i^s \text{ accepts maliciously}
  \right].
\]
Authentication security means that this probability is negligible. The post-quantum authentication definition uses the same malicious-acceptance event, but considers QPT adversaries.

Channel security captures confidentiality and integrity after acceptance. The adversary may interact with accepted sessions through encryption and decryption queries and eventually chooses a fresh accepted test session. Freshness means, informally, that the adversary has not revealed the test session key, has not revealed the key of its matching partner, and has not corrupted the relevant peer before acceptance. The test session contains a hidden bit $b$ that determines whether the encryption oracle returns real encryptions or challenge encryptions. If $\mathcal{A}$ outputs a guess $b'$, we write
\[
  p =
  \Pr\!\left[
    b'=b
    \,\middle|\,
    \pi_i^s \text{ is a fresh accepted test session}
  \right],
\]
and define
\[
  \mathsf{Adv}^{\mathsf{acce\mbox{-}mu\mbox{-}aenc}}_{\mathsf{SSH}}(\mathcal{A})
  =
  \left| p - \frac{1}{2} \right|.
\]
Channel security requires this advantage to be negligible for any efficient adversary.

For post-quantum channel security, the experiment models ``\textit{harvest now, decrypt later}'' attacks by restricting the handshake phase to honest executions. Instead of actively driving sessions with a send oracle, a QPT adversary obtains honestly generated transcripts through an execution oracle and may later issue reveal, corrupt, encryption and decryption queries subject to the same freshness conditions. Since these executions already have honest matching partners, the experiment focuses on confidentiality and integrity of the resulting channel. We denote the corresponding advantage by
\[
  \mathsf{Adv}^{\mathsf{pq\mbox{-}acce\mbox{-}aenc}}_{\mathsf{SSH}}(\mathcal{Q})
  =
  \left| p_{\mathsf{pq}} - \frac{1}{2} \right|,
\]
where $p_{\mathsf{pq}}$ is the probability that $\mathcal{Q}$ correctly guesses the hidden channel challenge bit for a fresh accepted session.

\section{KEM-Based User Authentication for SSH}
\label{sec:kem-auth}

This section presents \kemuauth, a KEM-based SSH user-authentication method that coexists with the existing signature-based \texttt{publickey} method. Within \kemuauth, the signature proof of possession normally used in public-key user authentication is replaced by a session-bound KEM response. As shown in Figure~\ref{fig:kem_auth}, \kemuauth is invoked only after the transport layer has established the encrypted channel and the session identifier. It therefore leaves SSH key exchange, host-key authentication, and connection-layer semantics unchanged, and confines the change to the user-authentication proof.

\begin{figure}[htbp]
\centering
\includegraphics[width=\columnwidth]{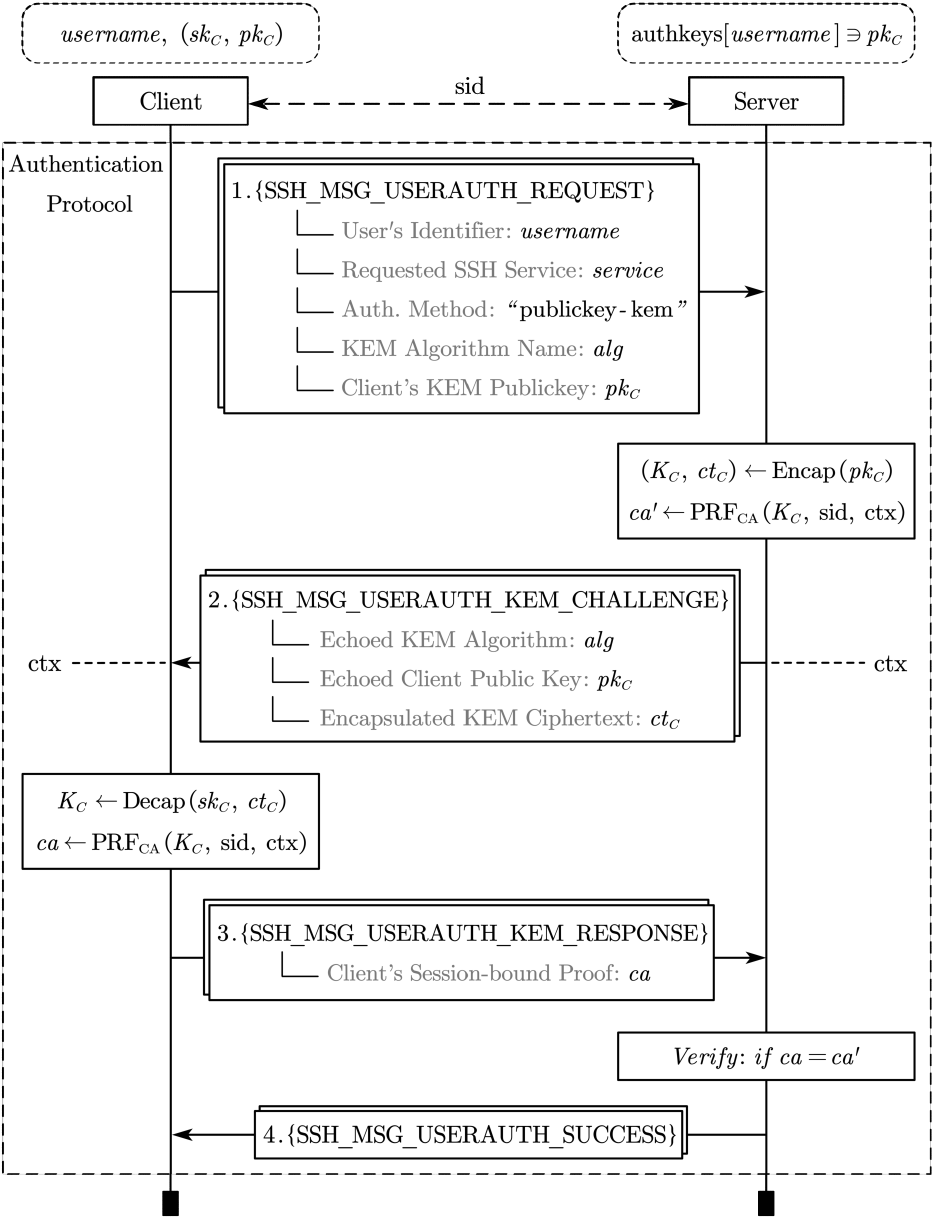}

\caption{KEM-based user authentication for SSH. Before this stage, the transport layer has established the session identifier $\mathsf{sid}$ and an encrypted channel. Messages in $\{\}$ are channel-protected. The server authorizes $pk_C$ for the claimed username, while the client holds the corresponding secret key $sk_C$. Grey text denotes field annotations rather than messages.}
\label{fig:kem_auth}
\end{figure}

At the start of user authentication, the client holds a long-term KEM key pair $(sk_C,pk_C)$, and the server maintains authorized public keys for each username, typically via an \texttt{authorized\_keys}-style database. The client begins by sending an encrypted \texttt{SSH\_MSG\_USERAUTH\_REQUEST} using the provisional method name \texttt{publickey-kem}. This request contains the username, requested service, KEM algorithm identifier $alg$, and the client's KEM public key $pk_C$. The server checks whether $pk_C$ is authorized for the claimed username and whether $alg$ is acceptable. If either check fails, the server returns \texttt{SSH\_MSG\_USERAUTH\_FAILURE}. Otherwise, it proceeds with the KEM challenge.

To challenge the client, the server runs
\[
  (K_C,ct_C) \leftarrow \mathsf{Encap}(pk_C)
\]
and computes the expected response
\[
  ca' \leftarrow \mathsf{PRF}_{\mathsf{CA}}(K_C,\mathsf{sid},\mathsf{ctx}).
\]
Let $\mathsf{msg}$ be the \kemuauth method transcript in Figure~\ref{fig:kem_auth}. It includes the initial request fields $(username,service,\texttt{publickey-kem}, alg,pk_C)$, the challenge fields $(alg,pk_C,ct_C)$, and the response message type, but not the response value $ca$. We define
\[
  \mathsf{ctx}=\mathsf{Encode}_{\mathsf{SSH}}(\mathsf{msg}).
\]
The server sends \texttt{SSH\_MSG\_USERAUTH\_KEM\_CHALLENGE}, which echoes $\mathsf{alg}$ and $pk_C$ and carries $ct_C$. Echoing these fields allows the client to match the challenge to the pending request encoded in $\mathsf{ctx}$ and avoids ambiguity across retried or concurrent authentication attempts.

Upon receiving the challenge, the client first checks that the echoed algorithm and public key match its pending authentication request. It then decapsulates
\[
  K_C \leftarrow \mathsf{Decap}(sk_C,ct_C).
\]
If decapsulation fails, the client rejects the challenge. Otherwise, it computes
\[
  ca \leftarrow \mathsf{PRF}_{\mathsf{CA}}(K_C,\mathsf{sid},\mathsf{ctx})
\]
and returns $ca$ in \texttt{SSH\_MSG\_USERAUTH\_KEM\_RESPONSE}. The server accepts if and only if $ca=ca'$, in which case it returns \texttt{SSH\_MSG\_USERAUTH\_SUCCESS}. Otherwise, it returns \texttt{SSH\_MSG\_USERAUTH\_FAILURE}. The explicit response is necessary because the server must verify possession of the decapsulated secret without feeding that secret into SSH transport key derivation.

We instantiate the client-authentication response function as
\[
  ca :=
  \mathsf{PRF}_{\mathsf{CA}}(K_C,\mathsf{sid},\mathsf{ctx})
  =
  \mathsf{HMAC}\!\left(
    K_C,\,
    \mathsf{label} \parallel
    \mathsf{sid} \parallel
    \mathsf{ctx}
  \right),
\]
where $\mathsf{label}$ is a fixed domain-separation label, such as \texttt{"ssh-publickey-kem-client-auth"}. The fields in $\mathsf{ctx}$ are encoded using SSH's length-prefixed binary string format, so the response is computed over an unambiguous representation of the authentication attempt. The concrete response length and hash function are specified by the advertised \kemuauth algorithm identifier.

The binding to $\mathsf{sid}$ and $\mathsf{ctx}$ gives the response the same session-binding role as the signature in the standard \texttt{publickey} method. The session identifier prevents replay across SSH sessions, while the authentication context binds the proof to the selected method, user, service, algorithm, public key, challenge ciphertext, and message flow. Since only the holder of $sk_C$ can recover $K_C$ from $ct_C$, a valid response proves possession of the authorized KEM secret key for this SSH authentication attempt. However, because the server also knows $K_C$ and can generate the same response, the resulting transcript provides deniable online authentication rather than transferable evidence of a client login. As a result, \kemuauth replaces the client signature with a KEM challenge--response proof, while requiring only one new authentication method and two method-specific user-authentication messages. Existing signature credentials and multi-method authentication policies remain available for incremental deployment. 

\section{Security Analysis}
\label{sec:security}

We now analyze SSH with \kemuauth, the KEM-based user-authentication method described in Section~\ref{sec:kem-auth}. Since \kemuauth leaves transport-layer key exchange, server host-key authentication, and channel-key derivation unchanged, the main new proof obligation is server acceptance of a client authenticated through the KEM-based response at the user-authentication layer. We first prove mutual ACCE authentication, then derive classical and post-quantum channel security by composing with the existing SSH ACCE analysis, and finally discuss the security requirements on $\mathsf{PRF}_{\mathsf{CA}}$.

\subsection{ACCE Authentication Security}
\label{subsec:auth-secure}

For the transport and server-authentication components, we adopt the assumptions of the post-quantum SSH ACCE analysis in~\cite{pqacce}: a collision-resistant hash function $\mathsf{H}$, an $\mathsf{EUF\text{-}CMA}$-secure digital signature scheme $\mathsf{DSS}$, a secure pseudorandom function $\mathsf{PRF}_{\mathsf{SSH}}$, a secure stateful authenticated encryption scheme $\mathsf{BSAE}$, and an $\mathsf{IND\text{-}CPA}$-secure ephemeral $\mathsf{KEM}_{\mathsf{e}}$ for transport-layer key exchange. Our client-authentication extension additionally requires the long-term client $\mathsf{KEM}_{\mathsf{C}}$ to be $\mathsf{IND\text{-}CCA}$ secure and the response function $\mathsf{PRF}_{\mathsf{CA}}$ to satisfy the standard PRF security notion.
Formal definitions of these primitives and security notions are collected in Appendix~\ref{app:defs}.

Our authentication definition follows the ACCE formulation of~\cite{pqacce}, but extends the malicious-acceptance event to mutual authentication. In~\cite{pqacce}, malicious acceptance is counted only for initiator sessions, since the analyzed mode authenticates the server to the client but not the client to the server. Since \kemuauth adds user-authentication-layer client authentication, we define
$\mathsf{Adv}_{\mathsf{SSH}}^{\mathsf{acce\text{-}mu\text{-}auth}}(\mathcal{A})$
as the probability that any initiator or responder session accepts believing that its peer is an uncorrupted party, but has no unique matching session of the opposite role.

\begin{lemma}[SSH is Mutual ACCE-Authentication-secure]
  \label{lem:mutual-auth}
  Consider SSH with signature-based server authentication and \kemuauth client authentication. Let $\mu$ be the length of the random cookies/nonces exchanged during transport-layer negotiation, let $n_P$ be the number of parties, let $n_S$ be the maximum number of sessions per party, and let $\ell$ be the output length of $\mathsf{PRF}_{\mathsf{CA}}$. For any PPT adversary $\mathcal{A}$, there exist adversaries $\mathcal{B}_1,\dots,\mathcal{B}_7$ such that
    \begin{align*}
    \mathsf{Adv}_{\mathsf{SSH}}^{\mathsf{acce\text{-}mu\text{-}auth}}(\mathcal{A})
    &\le \frac{(n_P n_S)^2}{2^{\mu}} + \mathsf{Adv}_{\mathsf{H}}^{\mathsf{coll}}(\mathcal{B}_1) \\
    \quad &+ n_P^2 n_S \Bigl( 
            \mathsf{Adv}_{\mathsf{DSS}}^{\mathsf{euf\text{-}cma}}(\mathcal{B}_2)
            + \mathsf{Adv}_{\mathsf{KEM}_{\mathsf{e}}}^{\mathsf{ind\text{-}cpa}}(\mathcal{B}_3) \\
    \quad &+ \mathsf{Adv}_{\mathsf{PRF}_{\mathsf{SSH}}}^{\mathsf{prf}}(\mathcal{B}_4)
            + \mathsf{Adv}_{\mathsf{BSAE}}^{\mathsf{bsae}}(\mathcal{B}_5) \\
    \quad &+ \mathsf{Adv}_{\mathsf{KEM}_{\mathsf{C}}}^{\mathsf{ind\text{-}cca}}(\mathcal{B}_6)
            + \mathsf{Adv}_{\mathsf{PRF_{CA}}}^{\mathsf{prf}}(\mathcal{B}_7) + 2^{-\ell} \Bigr).
    \end{align*}
\end{lemma}

\begin{proof}
The proof proceeds by game hopping. Each hop either aborts on a rare event or reduces a successful malicious acceptance to the security of an underlying primitive.
A complete game sequence is given in Appendix~\ref{app:acce}.

We first exclude collisions in the random values exchanged during transport-layer negotiation, contributing at most $(n_P n_S)^2/2^{\mu}$. We also abort on collisions in the exchange hash, which reduces to the collision resistance of $\mathsf{H}$. We then guess the first maliciously accepting session and, when needed, the peer identity involved in the impersonation; this accounts for the factor $n_P^2 n_S$ in the bound. The remaining analysis splits according to the role of the first accepting session.

\paragraph{Case A: malicious client acceptance}
This case is the server-authentication branch of the argument. After excluding nonce and hash collisions, a client that accepts without a matching server session must either accept a fresh valid host-key signature, yielding an $\mathsf{EUF\text{-}CMA}$ forgery against $\mathsf{DSS}$, or be led to accept through modified protected traffic. In the latter case, we replace the transport-layer KEM secret with a uniformly random value using the $\mathsf{IND\text{-}CPA}$ security of $\mathsf{KEM}_{\mathsf{e}}$, then replace the SSH key-schedule output with random channel keys using the PRF security of $\mathsf{PRF}_{\mathsf{SSH}}$. With independent random channel keys, any successful modification of protected user-authentication messages gives an adversary against the $\mathsf{BSAE}$ security of the channel encryption scheme. Thus, malicious client acceptance is bounded by the host-key signature, transport KEM, SSH key schedule, and channel-encryption terms.

\paragraph{Case B: malicious server acceptance}
This is the new authentication case introduced by our KEM-based user-authentication method. Here, a server accepts a client without a unique matching client session. The server's acceptance condition is a valid response
\[
  ca = \mathsf{PRF}_{\mathsf{CA}}(K_C,\mathsf{sid},\mathsf{ctx}),
\]
where $(K_C,ct_C) \leftarrow \mathsf{Encap}(pk_C)$ is generated by the server under the authorized long-term client KEM public key $pk_C$.

We guess the impersonated client identity and embed an $\mathsf{IND\text{-}CCA}$ challenge for $\mathsf{KEM}_{\mathsf{C}}$ into that client's long-term authentication key. The reduction uses its decapsulation oracle to simulate honest client decapsulations, except for the target challenge ciphertext used in the malicious-acceptance attempt. If the adversary can distinguish whether the target encapsulated key is real or random, this gives an $\mathsf{IND\text{-}CCA}$ adversary against $\mathsf{KEM}_{\mathsf{C}}$.

After the target encapsulated key is hidden, we replace the response derived from $\mathsf{PRF}_{\mathsf{CA}}$ with a uniformly random $\ell$-bit string, losing only the PRF advantage of $\mathsf{PRF}_{\mathsf{CA}}$. This replacement is applied at a fresh target input $y^\star=\mathsf{sid}^\star\parallel\mathsf{ctx}^\star$, whose collisions with non-target evaluations are already covered by the nonce and exchange-hash collision bounds. The expected response is therefore independent of the adversary's view before it sends $ca$. Therefore, the only remaining way to make the server accept is to guess this value, which succeeds with probability at most $2^{-\ell}$. Hence malicious server acceptance is bounded by the $\mathsf{IND\text{-}CCA}$ security of $\mathsf{KEM}_{\mathsf{C}}$, the PRF security of $\mathsf{PRF}_{\mathsf{CA}}$, and the direct-guessing term.

Combining the two cases with the initial collision and guessing losses gives the stated bound.
\end{proof}

We now consider authentication security against quantum adversaries. The malicious-acceptance event is unchanged: an initiator or responder session accepts believing that its peer is an uncorrupted party, but no unique matching session exists. Thus, for authentication security, the adversary class changes from PPT to QPT, while the interaction with the ACCE authentication experiment remains unchanged.

\begin{theorem}[Post-quantum ACCE authentication security]
  \label{thm:pq-auth}
  Consider SSH with signature-based server authentication and \kemuauth client authentication. Suppose that all primitive assumptions in Lemma~\ref{lem:mutual-auth} hold against quantum polynomial-time adversaries: collision resistance of $\mathsf{H}$, $\mathsf{EUF\text{-}CMA}$ security of $\mathsf{DSS}$, $\mathsf{IND\text{-}CPA}$ security of $\mathsf{KEM}_{\mathsf{e}}$, PRF security of $\mathsf{PRF}_{\mathsf{SSH}}$, authenticated-encryption security of $\mathsf{BSAE}$, $\mathsf{IND\text{-}CCA}$ security of $\mathsf{KEM}_{\mathsf{C}}$, and PRF security of $\mathsf{PRF}_{\mathsf{CA}}$. Then, for every QPT adversary $\mathcal{Q}$, there exist adversaries $\mathcal{B}_8,\dots,\mathcal{B}_{14}$ such that
  \begin{align*}
    \mathsf{Adv}_{\mathsf{SSH}}^{\mathsf{pq\text{-}acce\text{-}mu\text{-}auth}}(\mathcal{Q})
    &\le \frac{(n_P n_S)^2}{2^{\mu}} + \mathsf{Adv}_{\mathsf{H}}^{\mathsf{coll}}(\mathcal{B}_8) \\
    \quad +& n_P^2 n_S \Bigl( 
            \mathsf{Adv}_{\mathsf{DSS}}^{\mathsf{euf\text{-}cma}}(\mathcal{B}_9)
            + \mathsf{Adv}_{\mathsf{KEM}_{\mathsf{e}}}^{\mathsf{ind\text{-}cpa}}(\mathcal{B}_{10}) \\
    \quad +& \mathsf{Adv}_{\mathsf{PRF}_{\mathsf{SSH}}}^{\mathsf{prf}}(\mathcal{B}_{11})
            + \mathsf{Adv}_{\mathsf{BSAE}}^{\mathsf{bsae}}(\mathcal{B}_{12}) \\
    \quad +& \mathsf{Adv}_{\mathsf{KEM}_{\mathsf{C}}}^{\mathsf{ind\text{-}cca}}(\mathcal{B}_{13})
            + \mathsf{Adv}_{\mathsf{PRF_{CA}}}^{\mathsf{prf}}(\mathcal{B}_{14}) + 2^{-\ell} \Bigr).
    \end{align*}
\end{theorem}

\begin{proof}
The proof follows the game-hopping argument of Lemma~\ref{lem:mutual-auth}, with all primitive advantages interpreted against QPT adversaries. The reductions are straight-line and black-box, since they only simulate the ACCE authentication experiment, embed the relevant primitive challenge, and forward the adversary's output, without using random-oracle programmability, rewinding, the forking lemma, or measurements of the adversary's internal state. By the lifting argument of Song~\cite{song2014quantum}, the same sequence of hops applies in the quantum setting. Malicious client acceptance reduces to breaking the host-key signature, $\mathsf{KEM_e}$, $\mathsf{PRF_{SSH}}$, or channel encryption scheme. Malicious server acceptance reduces to breaking the $\mathsf{KEM_C}$ or $\mathsf{PRF}_{\mathsf{CA}}$, except with probability $2^{-\ell}$ for directly guessing the response. These quantities are negligible by assumption, which proves the claim.
\end{proof}

\subsection{ACCE Channel Security}
\label{subsec:channel-secure}

We now derive channel security for SSH with KEM-based user authentication. The proof follows the post-quantum SSH channel-security analysis of~\cite{pqacce}.\footnote{We include this statement to complete the ACCE treatment of our SSH variant and show how the channel-security bound composes with Lemma~\ref{lem:mutual-auth}.} This reuse is justified because our modification is confined to user authentication: all KEM-authentication messages are sent after \texttt{SSH\_MSG\_NEWKEYS}, inside the encrypted channel, and are not inputs to transport key exchange, session-identifier derivation, or channel-key derivation. Thus, the classical ACCE bound only replaces server-only authentication with the mutual-authentication term from Lemma~\ref{lem:mutual-auth}.

\begin{lemma}[Channel security, KEM-based user auth. mode]
  \label{lem:channel}
  Let $\mu$ be the length of the random cookies/nonces exchanged during transport-layer negotiation, let $n_P$ be the number of parties, and let $n_S$ be the maximum number of sessions per party. For any PPT adversary $\mathcal{A}$, there exist adversaries $\mathcal{B}_{15},\mathcal{B}_{16},\mathcal{B}_{17}$ such that
    \begin{align*}
        \mathsf{Adv}_{\mathsf{SSH}}^{\mathsf{acce\text{-}mu\text{-}aenc}}(\mathcal{A})
        &\le \mathsf{Adv}_{\mathsf{SSH}}^{\mathsf{acce\text{-}mu\text{-}auth}}(\mathcal{A}) \\
        \quad &+ n_P^2 n_S^2 \bigl( 
            \mathsf{Adv}_{\mathsf{KEM_e}}^{\mathsf{ind\mbox{-}cpa}}(\mathcal{B}_{15}) \\
        \quad &+ \mathsf{Adv}_{\mathsf{PRF}_{\mathsf{SSH}}}^{\mathsf{prf}}(\mathcal{B}_{16})
            + \mathsf{Adv}_{\mathsf{BSAE}}^{\mathsf{bsae}}(\mathcal{B}_{17}) \bigr).
    \end{align*}
\end{lemma}

\begin{proof}
The proof aborts if the test session accepts maliciously, an event bounded by $\mathsf{Adv}_{\mathsf{SSH}}^{\mathsf{acce\text{-}mu\text{-}auth}}$. Conditioned on no such event, the test session has a unique matching partner. The proof then guesses the test session and partner, replaces the transport KEM secret using $\mathsf{IND\text{-}CPA}$ security of $\mathsf{KEM}_{\mathsf{e}}$, replaces the key-schedule output of $\mathsf{PRF}_{\mathsf{SSH}}$, and embeds the channel into a $\mathsf{BSAE}$ challenger, yielding the stated bound.
\end{proof}

For post-quantum channel security, the KEM-authentication layer does not affect the proof. In the ``\textit{harvest now, decrypt later}'' experiment, sessions obtained through $\mathsf{OExecute}(i,s,j,t)$ are honestly generated and already have matching partners. The argument therefore depends only on the transport-layer KEM, the SSH key schedule, and the channel encryption scheme.

\begin{theorem}[Post-quantum channel security]
  \label{thm:pq-channel}
  Let $n_P$ be the number of parties and $n_S$ the maximum number of sessions per party. For any QPT adversary $\mathcal{Q}$, there exist QPT adversaries $\mathcal{B}_{18},\mathcal{B}_{19},\mathcal{B}_{20}$ such that
    \begin{align*}
        \mathsf{Adv}_{\mathsf{SSH}}^{\mathsf{pq}\text{-}\mathsf{acce}\text{-}\mathsf{aenc}}(\mathcal{Q})
        \le & n_P n_S \bigl( \mathsf{Adv}_{\mathsf{KEM_e}}^{\mathsf{ind\mbox{-}cpa}}(\mathcal{B}_{18}) \\
        +& \mathsf{Adv}_{\mathsf{PRF}_{\mathsf{SSH}}}^{\mathsf{prf}}(\mathcal{B}_{19}) + \mathsf{Adv}_{\mathsf{BSAE}}^{\mathsf{bsae}}(\mathcal{B}_{20}) \bigr).
    \end{align*}
\end{theorem}

\begin{proof}
Since execution-oracle sessions are honest, no malicious-acceptance abort is needed. The proof guesses the test session, replaces the $K_e$ and SSH key-schedule output using the quantum security of $\mathsf{KEM}_{\mathsf{e}}$ and $\mathsf{PRF}_{\mathsf{SSH}}$, and embeds the resulting channel into a quantum-secure $\mathsf{BSAE}$ challenger. KEM-authentication messages occur only after channel keys are established and do not affect these game hops.
\end{proof}

\subsection{Security Role of $\mathsf{PRF}_{\mathsf{CA}}$}
\label{subsec:prfca}

The role of $\mathsf{PRF}_{\mathsf{CA}}$ is to turn the KEM shared secret $K_C$ into an explicit authentication value that is bound to the SSH session and the concrete user-authentication attempt. In the malicious server-acceptance case of Lemma~\ref{lem:mutual-auth}, the IND-CCA security of the client KEM first hides the challenge secret $K_C$ from the adversary. The PRF security of $\mathsf{PRF}_{\mathsf{CA}}$ then allows the expected response on $(\mathsf{sid},\mathsf{ctx})$ to be replaced by a uniformly random $\ell$-bit string, leaving only the direct guessing probability $2^{-\ell}$. For post-quantum authentication security, the same requirement is interpreted against QPT distinguishers. Our HMAC-based instantiation follows Section~\ref{sec:kem-auth} and relies on standard PRF assumptions for HMAC~\cite{rfc2104,fips198,bellare2006hmac}.

\section{Implementation and Evaluation}
\label{sec:evaluation}

This section evaluates \kemuauth as a practical drop-in user-authentication method for SSH. We study whether replacing client signatures with a KEM-based proof changes the cost profile of post-quantum SSH authentication in realistic deployment settings. Our evaluation covers an OpenSSH-based implementation, primitive-level authentication costs, end-to-end handshake latency under representative RTTs, TCP initial-window effects, and server-side online authentication cost. The results show that \kemuauth is competitive with signature-based public-key authentication in common network settings, while providing clearer benefits when large post-quantum signatures stress transmission or server-side verification.
The raw experimental data are reported in Appendix~\ref{app:raw-eval}.

\subsection{Instantiation and Implementation}

We implemented \kemuauth in an OQS-v10 fork based on OpenSSH 10.2p1 as a local extension to the SSH user-authentication layer. The prototype adds a new authentication method, \texttt{publickey-kem}, for KEM public-key credentials, allowing it to coexist with existing SSH authentication methods. It uses liboqs~\cite{liboqs} for post-quantum KEM operations and currently supports ML-KEM-512, ML-KEM-768, and ML-KEM-1024 as client credential algorithms. Following the transcript binding in Section~\ref{sec:kem-auth}, the client-authentication response is instantiated as an HMAC-SHA-2 value keyed by the decapsulated KEM secret and computed over the fixed label \texttt{ssh-publickey-kem-client-auth}, the SSH session identifier, and the canonical SSH encoding of the \kemuauth authentication context. KEM choices and response-function variants are exposed through \kemuauth algorithm identifiers $alg$ for future extensibility. The code, scripts, and patches used in our evaluation are publicly available.\footnote{Our implementation is available at \url{https://github.com/bhivfsadh/KEM-Based-User-Authentication-for-Post-Quantum-SSH}}

\subsection{Experimental Setup}
\label{subsec:exp-setup}

\textbf{Testbed.} All experiments were conducted on a workstation-class test host equipped with an Intel Core Ultra 9 275HX at 2.7GHz, configured to use 16 hardware threads with one thread per core, and 32GB RAM, running Ubuntu 22.04.5 LTS. To reduce measurement noise, we disabled CPU frequency scaling and fixed the CPU governor to \texttt{performance}. The code was compiled with GCC~11.4 using \texttt{-O2 -mavx2}. Post-quantum primitives were provided by liboqs~0.15.0, including AVX2-optimized implementations of ML-KEM, ML-DSA, Falcon, and SLH-DSA. Classical cryptographic algorithms, including Ed25519, were provided by OpenSSL~3.0.2.

\textbf{Network configuration.} Unless otherwise specified, the SSH transport-layer key exchange was fixed to \texttt{mlkem768x25519-sha256}, and the TCP maximum segment size (MSS) was set to 1460 bytes. We used RTTs of about 37\,ms, 67\,ms, and 163\,ms as representative low-, intermediate-, and high-latency settings for the main evaluation~\cite{assesstlsssh}. For sensitivity analysis, we swept RTT from 0 to 200\,ms and evaluated random packet-loss rates from 0\% to 2\%. To study TCP transmission effects, we varied the TCP initial congestion window from 3 to 50 MSS, spanning constrained small-window settings, the common IW10 setting, and larger windows observed in deployed systems~\cite{rfc9006,rfc6928,rueth2018demystifying}. All other experiments used 10 MSS. Network delay and TCP parameters were configured using Linux \texttt{tc/netem} and \texttt{ip route}.

\textbf{Measurement methodology.} We use two modes. For cryptographic microbenchmarks, each operation was run for 1,000 warm-up and 10,000 measured iterations, with five repetitions per experiment. We report median operation time measured using \texttt{clock\_gettime(CLOCK\_MONOTONIC)}. For end-to-end tests, we measure from the TCP SYN to successful user authentication over 5,000 handshakes per configuration. After each handshake, we close the connection and clear client-side cached state, including \texttt{known\_hosts}, to avoid repeated-connection artifacts. We report latency in milliseconds, emphasizing medians and tail latency in the main text.
Raw distributions are reported in Appendix~\ref{app:raw-eval}.

\subsection{Microbenchmark Characterization}

\begin{table*}[htbp]
\centering
\caption{Sizes and operation costs of authentication primitives used in our evaluation}
\label{tab:alg-bench}
\scriptsize
\setlength{\tabcolsep}{4pt}
\renewcommand{\arraystretch}{1.3}
\newcommand{\thickhline}{\noalign{\hrule height 0.8pt}}
\newcommand{\spacedhline}{\noalign{\vskip 1.5pt}\hline\noalign{\vskip 1.5pt}}

\begin{tabular}{l c| c r| r r| r r| c}
\thickhline
\noalign{\vskip 2pt}
Authentication & Notation & Problem & NIST PQ & Public Key & Signature/Ciphertext & Sign/Encaps. & Verify/Decaps. & Optimization \\
Algorithm & & Family & Category & (Bytes) & (Bytes) & (ms) & (ms) & \\
\spacedhline

Ed25519~\cite{rfc8709} & ed25519 & ECC Discrete Logarithm & $\approx$0 bits & 32 & 64 & 0.019 & 0.047 & OpenSSL \\
\spacedhline

ML-DSA-44~\cite{fips204} & md44 & Module-LWE/MSIS & Level 2 & 1312 & 2420 & 0.033 & 0.014 & AVX2 \\
ML-DSA-65~\cite{fips204} & md65 & Module-LWE/MSIS & Level 3 & 1952 & 3309 & 0.054 & 0.022 & AVX2 \\
ML-DSA-87~\cite{fips204} & md87 & Module-LWE/MSIS & Level 5 & 2592 & 4627 & 0.066 & 0.035 & AVX2 \\
\spacedhline

Falcon-512~\cite{oqsfalcon} & falcon512 & NTRU-SIS & Level 1 & 897 & 752 & 0.131 & 0.021 & AVX2 \\
Falcon-1024~\cite{oqsfalcon} & falcon1024 & NTRU-SIS & Level 5 & 1793 & 1462 & 0.255 & 0.042 & AVX2 \\
\spacedhline

SLH-DSA-SHA2-128f~\cite{fips205} & sd128f & Hash-based & Level 1 & 32 & 17088 & 12.704 & 0.734 & AVX2 \\
SLH-DSA-SHA2-192f~\cite{fips205} & sd192f & Hash-based & Level 3 & 48 & 35664 & 21.380 & 1.122 & AVX2 \\
SLH-DSA-SHA2-256f~\cite{fips205} & sd256f & Hash-based & Level 5 & 64 & 49856 & 43.370 & 1.124 & AVX2 \\
\spacedhline

ML-KEM-512~\cite{fips203} & mk512 & Module-LWE & Level 1 & 800 & 768 & 0.005 & 0.006 & AVX2 \\
ML-KEM-768~\cite{fips203} & mk768 & Module-LWE & Level 3 & 1184 & 1088 & 0.008 & 0.009 & AVX2 \\
ML-KEM-1024~\cite{fips203} & mk1024 & Module-LWE & Level 5 & 1568 & 1568 & 0.010 & 0.012 & AVX2 \\
\noalign{\vskip 2pt}
\thickhline

\end{tabular}
\end{table*}

Table~\ref{tab:alg-bench} summarizes the primitive measurements used to interpret the protocol-level results. We select standardized or SSH-relevant baselines, with Ed25519 as the classical reference, and group post-quantum parameter sets by NIST category only to coarsely align target security strength, not to imply equivalent security. We report public-key and signature or ciphertext sizes together with signing and verification or encapsulation and decapsulation times, omitting offline key generation for long-term SSH credentials. For Falcon, whose signatures are variable-length, we report the maximum signature size. Object sizes help explain packetization effects, while local operation costs inform the server-side analysis. However, because the post-quantum schemes use AVX2-optimized liboqs implementations whereas Ed25519 uses OpenSSL, these measurements should not be interpreted as a general ranking of primitive families.

\subsection{End-to-End Latency Evaluation}

We next evaluate how \kemuauth affects end-to-end SSH handshake latency, with the transport key exchange fixed to \texttt{mlkem768x25519-sha256}. We first isolate the client-authentication algorithm, then study practical migration configurations across RTT regimes.

\begin{figure}[htbp]
\centering
\includegraphics[width=\columnwidth]{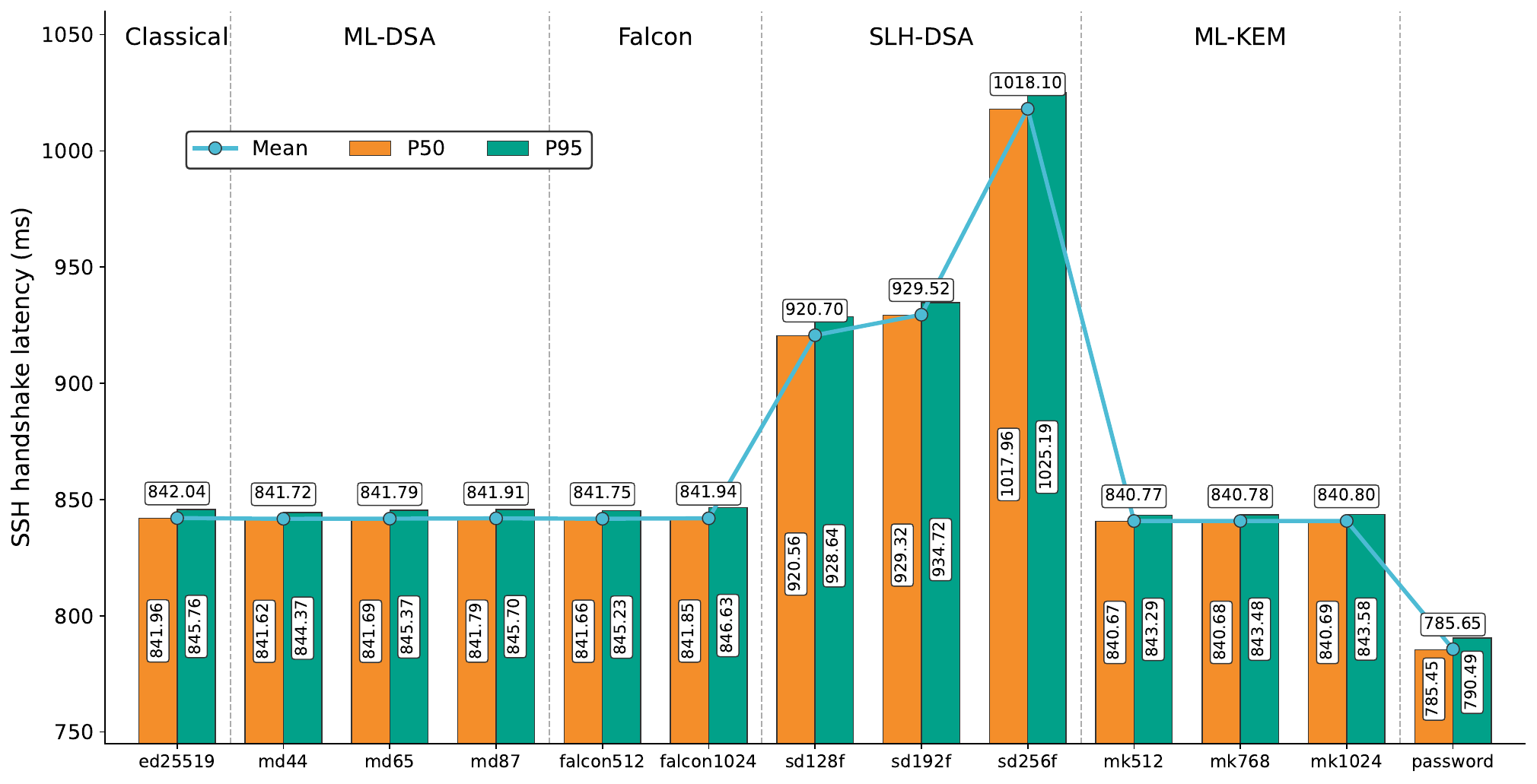}
\caption{End-to-end SSH handshake latency when varying only the client-authentication algorithm at RTT~$\approx$~67\,ms. Bars show the median and 95th percentile, and the line shows the mean.}
\label{fig:client-auth-latency}
\end{figure}

Figure~\ref{fig:client-auth-latency} isolates the effect of client authentication, with Ed25519 server authentication, \texttt{mlkem768x25519-sha256} transport key exchange, and RTT fixed at approximately 67\,ms. ML-KEM remains in the same latency band as compact signatures: ML-KEM-768 completes at 840.68\,ms median latency, while Ed25519, ML-DSA, and Falcon remain within about 2\,ms, indicating that fixed protocol and network costs dominate primitive-level differences. Password authentication is faster at 785.45\,ms because it avoids the public-key probe or KEM challenge, but follows a distinct credential model rather than public-key proof of possession. In contrast, SLH-DSA-SHA2-192f and SLH-DSA-SHA2-256f reach 929.32\,ms and 1017.96\,ms, respectively, with the latter reflecting higher signing cost and an extra TCP flight under the default 10 MSS setting.

\begin{figure*}[t]
\centering
\includegraphics[width=\textwidth]{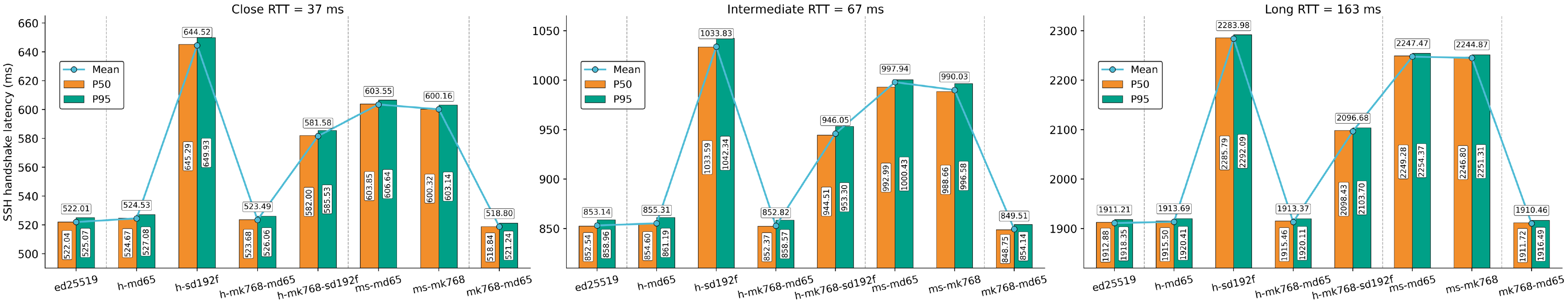}
\caption{End-to-end SSH handshake latency under different RTT regimes. Bars show the median and 95th percentile, and the line shows the mean. Labels without a separator indicate that the client and server use the same authentication primitive. The prefix \texttt{h-} denotes a hybrid configuration that combines Ed25519 with the indicated post-quantum authenticator. The prefix \texttt{ms-} denotes SSH-native multi-step client authentication in which Ed25519 is followed by the indicated post-quantum method, with server authentication fixed to \texttt{md65 + ed25519}. For labels of the form \texttt{X-Y}, \texttt{X} denotes the client-authentication method and \texttt{Y} denotes the server-authentication method.}
\label{fig:rtt-latency}
\end{figure*}

Figure~\ref{fig:rtt-latency} evaluates practical migration configurations across close, intermediate, and long RTT regimes, with the TCP initial window fixed at 10 MSS. These include hybrid configurations combining a classical and a post-quantum authenticator, pure post-quantum modes, and SSH-native multi-step configurations running two user-authentication methods sequentially. As expected, absolute handshake latency is largely driven by RTT, so the lower primitive cost of ML-KEM relative to ML-DSA is mostly absorbed by fixed network and protocol costs. In this setting, \kemuauth closely matches ML-DSA-based migration paths: \texttt{h-mk768-md65} is up to about 3\,ms faster than the ML-DSA hybrid baseline \texttt{h-md65} across all RTTs. The benefit is larger when client authentication uses large SLH-DSA signatures. Both \texttt{h-sd192f} and \texttt{h-mk768-sd192f} use the same SLH-DSA server authenticator, while the latter replaces client-side SLH-DSA with ML-KEM. This reduces median handshake latency by 9.8\%, 8.6\%, and 8.2\% at close, intermediate, and long RTTs, respectively. Multi-step configurations are slower than parallel hybrids because they add a serial authentication stage. The RTT-sweep and packet-loss experiments confirm that these trends remain stable under broader network conditions.
Detailed results are reported in Tables~\ref{tab:rtt-sensitivity} and~\ref{tab:loss-sensitivity}.

\subsection{Transport Effects of Authentication Object Size}

We next isolate the transmission effect of authentication-object size. The RTT is fixed at approximately 67\,ms, and the TCP initial congestion window is varied from 3 to 50 MSS. This setting stresses whether public keys, signatures, and KEM ciphertexts fit into the early TCP flight, which is especially relevant for constrained or conservative TCP configurations.

\begin{figure}[htbp]
\centering
\includegraphics[width=\columnwidth]{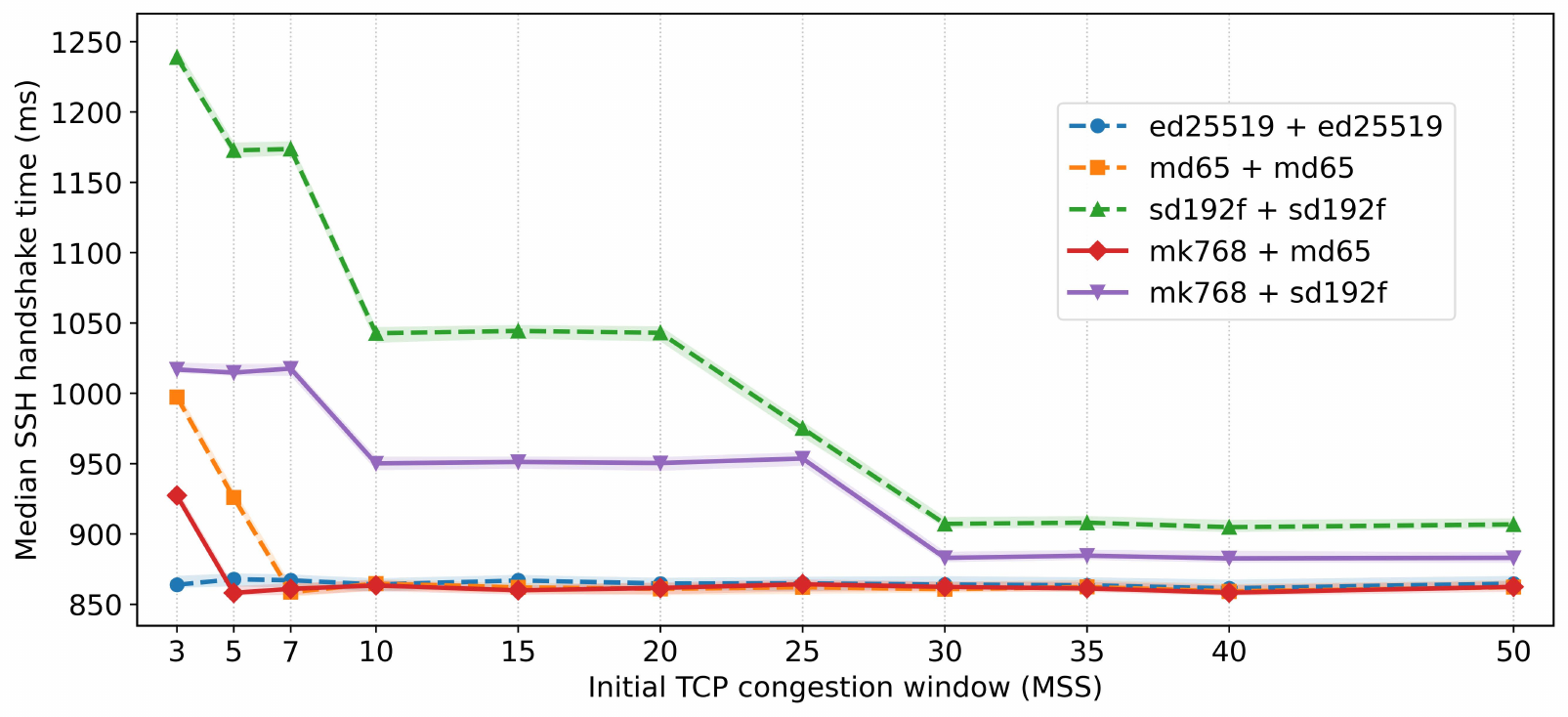}
\caption{SSH handshake latency under varying TCP initial congestion windows at RTT~$\approx$~67\,ms, with the transport key exchange fixed to \texttt{mlkem768x25519-sha256}. Lines show median latency and shaded bands show the 5th--95th percentile range. \texttt{X + Y} denotes client authentication with \texttt{X} and server authentication with \texttt{Y}.}
\label{fig:initcwnd}
\end{figure}

Figure~\ref{fig:initcwnd} shows that compact authentication configurations are relatively stable across window sizes, while configurations with large post-quantum signatures become sensitive to small initial windows. The Ed25519 baseline and \texttt{mk768 + md65} remain in the same latency band, whereas \texttt{md65 + md65} and especially \texttt{sd192f + sd192f} incur clear penalties at 3 and 5 MSS. Replacing the client signature with ML-KEM reduces this sensitivity. For ML-DSA, \texttt{mk768 + md65} lowers median latency by 7.0\% at 3 MSS and 7.3\% at 5 MSS compared with \texttt{md65 + md65}. At 7 MSS and above, the difference remains within 0.3\% with no consistent direction. For SLH-DSA, \texttt{mk768 + sd192f} reduces median latency by 17.9\% at 3 MSS and 13.5\% at 5 MSS compared with \texttt{sd192f + sd192f}, and the reduction remains 8.9\% at the default 10 MSS. Overall, \kemuauth mainly benefits ML-DSA configurations under small initial windows, while the larger SLH-DSA signatures yield latency savings over a wider range of window sizes.

\subsection{Server-Side Online Cryptographic Cost}

We finally evaluate the server-side online cryptographic cost of client authentication and quantify how the resulting savings affect connection throughput under concurrent load. The measurements connect per-credential cryptographic processing with end-to-end server behavior as concurrency increases.

\begin{table}[htbp]
\centering
\caption{Per-Credential Online Cryptographic Cost}
\label{tab:server-cost}
\scriptsize
\setlength{\tabcolsep}{5pt}
\renewcommand{\arraystretch}{1.3}
\newcommand{\thickhline}{\noalign{\hrule height 0.8pt}}
\newcommand{\spacedhline}{\noalign{\vskip 1.5pt}\hline\noalign{\vskip 1.5pt}}

\begin{tabular}{l c| c| r r}
\thickhline
\noalign{\vskip 2pt}
Client Auth. & NIST PQ & Server-side & Median Cost & Relative to \\
Notation & Category & Operation & (ms) & \texttt{mk768} \\
\spacedhline

ed25519 & $\approx$0 bits & Verify & 0.047 & 5.2$\times$ \\
\spacedhline

md44 & Level 2 & Verify & 0.014 & 1.6$\times$ \\
md65 & Level 3 & Verify & 0.022 & 2.4$\times$ \\
md87 & Level 5 & Verify & 0.035 & 3.9$\times$ \\
\spacedhline

sd128f & Level 1 & Verify & 0.734 & 81.6$\times$ \\
sd192f & Level 3 & Verify & 1.122 & 124.7$\times$ \\
sd256f & Level 5 & Verify & 1.124 & 124.9$\times$ \\
\spacedhline

mk512 & Level 1 & Encaps. + HMAC & 0.007 & 0.8$\times$ \\
mk768 & Level 3 & Encaps. + HMAC & 0.009 & 1.0$\times$ \\
mk1024 & Level 5 & Encaps. + HMAC & 0.012 & 1.3$\times$ \\
\noalign{\vskip 2pt}
\thickhline

\end{tabular}
\end{table}

Table~\ref{tab:server-cost} shows that \kemuauth has lower per-credential server-side cryptographic cost than all evaluated signature-verification baselines. Within NIST category 3, ML-KEM-768 plus HMAC costs 0.009\,ms, 59.1\% below ML-DSA-65 verification at 0.022\,ms. The larger gaps against the SLH-DSA variants mainly reflect the high verification cost of hash-based signatures rather than a \kemuauth-specific advantage.

With transport key exchange and server authentication fixed, we compared \kemuauth with ML-DSA-65 in a concurrent-load experiment using \(N \in \{1,8,16,32,64\}\) client workers. \kemuauth improves the connection-completion rate by 4.5\%, 3.3\%, 1.4\%, 1.1\%, and 0.35\%, respectively. These gains reflect its lower authentication-side cryptographic cost and narrow near CPU saturation, where common SSH costs such as key exchange, packet processing, process management, and scheduling increasingly govern the overall connection rate.
Appendix~\ref{app:concurrent-eval} reports the detailed experimental setup and complete throughput, latency, and memory results.

\section{Discussion and Future Work}

\textbf{Client-side KEM Auth. vs. Mutual KEM Auth.}
\kemuauth applies KEM authentication at the user-authentication layer because SSH server authentication is embedded in the transport layer and coupled with key exchange, session binding, and host-key trust. Extending KEM authentication to the server would require broader changes and could add messages, increasing deployment complexity and eroding its benefits, especially on high-latency links. \kemuauth therefore uses SSH's method-extensible client-authentication layer as a clean integration point without redesigning transport semantics or adding protocol overhead.

\textbf{Deniable Authentication and Auditability.}
\kemuauth provides deniable online authentication through a session-bound response demonstrating possession of the client KEM secret key without creating transferable evidence. Since the server creates the challenge ciphertext and obtains the shared secret, it can reproduce the response and simulate a valid transcript. Third parties cannot attribute the transcript uniquely to the client, consistent with prior KEM-authenticated and deniable-authentication designs~\cite{kemtls,diraimondo2006deniable}. \kemuauth suits online access control, while deployments requiring non-repudiation or independently verifiable audit evidence should retain signature authentication alone or combine it with \kemuauth through SSH's multi-method policies~\cite{nist_digitalsignature,krawczyk2003sigma}.

\textbf{Implementation and Side-Channel Security.} 
Our ACCE analysis establishes protocol-level security under the IND-CCA assumption for the client KEM, while leakage from concrete implementations lies outside this model. For example, ciphertexts chosen by a malicious server may expose the client's long-term decapsulation interface as a potential decapsulation-oracle surface. Our public artifact includes preliminary malformed-ciphertext and timing checks that found no evident anomalies, but these do not constitute a comprehensive implementation-security evaluation. Deployment therefore additionally relies on constant-time decapsulation and broader side-channel and fault-injection assessment.

\textbf{Toward Standardization and Interoperability.}
Turning \kemuauth into an interoperable SSH extension requires a precise wire specification for method names, message formats, KEM key encodings, algorithm identifiers, transcript binding, and downgrade behavior, together with deployment integration for \texttt{authorized\_keys}, SSH agents, and multi-method policies. Since \kemuauth runs within the established SSH channel, downgrade handling primarily concerns policy fallback to other enabled authentication methods rather than tampering within the method. The prototype and controlled measurements establish feasibility, while broader adoption requires interoperability testing, machine-checked validation of the protocol state machine against implementations, and standardization within the post-quantum SSH ecosystem.

\section{Conclusion}

We presented \kemuauth, a drop-in KEM-based user-authentication method for post-quantum SSH that replaces client signatures with a session-bound KEM proof while preserving SSH's public-key credential model and authentication flow. We proved its security in the post-quantum ACCE framework and implemented it in OpenSSH with liboqs. Our evaluation shows competitive end-to-end performance across practical migration settings, with particular benefits when post-quantum signatures increase TCP transmission or server-side verification costs. Overall, \kemuauth identifies a practical, lightweight middle ground for post-quantum SSH authentication, extending KEMs beyond transport-layer key exchange without redesigning SSH's transport semantics.

\bibliographystyle{IEEEtran}
\bibliography{refs}

\appendices

\section{CRYPTOGRAPHIC DEFINITIONS}\label{app:defs}

The security proof of our proposed scheme relies on the following cryptographic primitives and their standard security notions.

\begin{definition}[Hash Function and Collision Resistance]
  A hash function $\mathsf{H} : \{0,1\}^* \to \{0,1\}^\lambda$ maps an arbitrary-length bit string to a fixed-length output of $\lambda$ bits. Collision resistance requires that for any probabilistic polynomial-time (PPT) adversary $\mathcal{A}$ the probability that $\mathcal{A}$ outputs two distinct messages $m \neq m'$ such that $\mathsf{H}(m) = \mathsf{H}(m')$ is negligible. This probability is denoted by
  $$
  \mathsf{Adv}_{\mathsf{H}}^{\mathsf{coll}}(\mathcal{A}) = \Pr\left[ (m,m') \leftarrow \mathcal{A} : m \neq m' \land \mathsf{H}(m) = \mathsf{H}(m') \right].
  $$
\end{definition}

\begin{definition}[Pseudorandom Function (PRF)]
  A pseudorandom function family $\mathsf{PRF} : \mathcal{K} \times \mathcal{X} \to \{0,1\}^\lambda$ maps a key $k \in \mathcal{K}$ and an input $x \in \mathcal{X}$ to a fixed-length output. PRF security is defined via the following experiment: the challenger chooses a random key $k \leftarrow \mathcal{K}$ and a random bit $b \in \{0,1\}$. If $b = 0$, the adversary $\mathcal{A}$ is given oracle access to $\mathsf{PRF}(k,\cdot)$; if $b = 1$, $\mathcal{A}$ is given access to a truly random function $R : \mathcal{X} \to \{0,1\}^\lambda$. $\mathcal{A}$ may adaptively query the oracle and finally outputs a guess $b'$. The adversary wins if $b' = b$, and its advantage is defined as
  $$
  \mathsf{Adv}_{\mathsf{PRF}}^{\mathsf{prf}}(\mathcal{A}) = \left|\Pr[b' = b] - \frac12\right|.
  $$
\end{definition}

\begin{definition}[Stateful Authenticated Encryption (BSAE)]
  A $\mathsf{BSAE}$ scheme consists of the following algorithms:
  \begin{itemize}
    \item $\mathsf{StInit}()$ outputs an initial state $st$;
    \item $\mathsf{Enc}(k, m, st)$ takes a key $k$, a plaintext $m$, and the current state $st$, and outputs a ciphertext $c$ together with an updated state $st'$;
    \item $\mathsf{Dec}(k, c, st)$ takes a key $k$, a ciphertext $c$, and the current state $st$, and outputs either a plaintext $m$ or the rejection symbol $\perp$, along with an updated state $st'$.
  \end{itemize}
  Security comprises two aspects:
  \begin{itemize}
    \item Ciphertext indistinguishability: an adversary cannot distinguish ciphertexts corresponding to two equal-length plaintexts, even when allowed to adaptively choose plaintexts and observe ciphertexts (while taking state updates into account). This is typically denoted as IND-BSAE.
    \item Ciphertext integrity: an adversary cannot forge a new valid ciphertext, i.e., cannot cause the decryption algorithm to output a non-$\perp$ value that was not obtained from the encryption oracle. This is typically denoted as INT-BSAE.
  \end{itemize}
  The combined security (satisfying both properties) is denoted by $\mathsf{Adv}_{\mathsf{BSAE}}^{\mathsf{bsae}}(\mathcal{A})$. For full details of the security experiment we refer to Bergsma et al.~\cite{bellare2002authenticated}.
\end{definition}

\begin{definition}[Digital Signature (EUF-CMA)]
  A signature scheme $\mathsf{DSS} = (\mathsf{Gen}, \mathsf{Sign}, \mathsf{Vrfy})$ consists of key generation, signing, and verification algorithms. Existential unforgeability under chosen message attacks (EUF-CMA) is defined via the following experiment: the challenger generates a key pair $(pk, sk) \leftarrow \mathsf{Gen}(1^\lambda)$ and gives $pk$ to the adversary $\mathcal{A}$. The adversary may adaptively query a signing oracle $\mathsf{Sign}(sk, \cdot)$ and obtains signatures on chosen messages. Eventually, $\mathcal{A}$ outputs a message $m^*$ and a signature $\sigma^*$. The adversary wins if $\mathsf{Vrfy}(pk, m^*, \sigma^*) = 1$ and $m^*$ was never submitted to the signing oracle. The advantage is defined as
  $$
  \mathsf{Adv}_{\mathsf{DSS}}^{\mathsf{euf\mbox{-}cma}}(\mathcal{A}) = \Pr[\mathcal{A} \text{ outputs a valid forgery}].
  $$
\end{definition}

\begin{definition}[Key Encapsulation Mechanism (KEM)]
  A KEM consists of three algorithms:
  \begin{itemize}
    \item $\mathsf{KeyGen}(1^\lambda)$ outputs a public key $pk$ and a private key $sk$;
    \item $\mathsf{Encap}(pk)$ outputs a shared secret $K$ and a ciphertext $ct$;
    \item $\mathsf{Decap}(sk, ct)$ outputs a shared secret $K$ or the rejection symbol $\perp$.
  \end{itemize}
  Correctness requires that for all $(pk,sk)\leftarrow \mathsf{KeyGen}(1^\lambda)$ and $(K,ct)\leftarrow \mathsf{Encap}(pk)$, we have $\Pr[\mathsf{Decap}(sk,ct)=K]\ge 1-\delta$, where $\delta$ is negligible.

  Two security notions are standard:
  \begin{itemize}
    \item IND-CPA security: The challenger generates $(pk, sk)$, computes $(K_0, ct^*) \leftarrow \mathsf{Encap}(pk)$, and picks uniformly random $K_1$. It then chooses a random bit $b \in \{0,1\}$ and gives $(pk, ct^*, K_b)$ to the adversary $\mathcal{A}$. $\mathcal{A}$ outputs a guess $b'$ and wins if $b' = b$. The advantage is defined as
    $$
    \mathsf{Adv}_{\mathsf{KEM}}^{\mathsf{ind\mbox{-}cpa}}(\mathcal{A}) = \left|\Pr[b' = b] - \frac12\right|.
    $$
    In this game the adversary has no access to a decapsulation oracle.
    \item IND-CCA security: The game is the same as IND-CPA, except that the adversary may additionally adaptively query a decapsulation oracle $\mathsf{Decap}(sk, \cdot)$ (with the restriction that it cannot query $ct^*$). The adversary outputs a guess $b'$ and wins if $b' = b$. The advantage is defined as
    $$
    \mathsf{Adv}_{\mathsf{KEM}}^{\mathsf{ind\mbox{-}cca}}(\mathcal{A}) = \left|\Pr[b' = b] - \frac12\right|.
    $$
  \end{itemize}
\end{definition}

\section{Reductionist Security Analysis}
\label{app:acce}

This appendix provides a formal security model for the SSH protocol extended with KEM‑based client authentication, following the ACCE (Authenticated and Confidential Channel Establishment) framework~\cite{acce,pqacce}. The model captures both classical and post‑quantum adversaries, the latter via the $\mathsf{OExecute}$ oracle for “\textit{harvest now, decrypt later}” attacks.

\subsection{Model Syntax}

Let $n_P$ be the number of parties and denote the set of parties by $P_1,\dots,P_{n_P}$.
Each party $P_i$ generates a long-term key pair $(sk_i,pk_i)$ using a key generation algorithm $\mathsf{KeyGen}$.
In the server-only authentication part of the handshake, the server holds a signing key; for client authentication, the client holds a static KEM key pair.
All public keys are assumed to be available to the adversary.

A session is a single protocol execution at a party.  
The $s$-th session at party $P_i$ is denoted $\pi_i^s$.

\begin{definition}[Per-session variables]
  Each session maintains the following variables:
  \begin{itemize}
    \item $\rho \in \{\mathsf{init},\mathsf{resp}\}$: the role of the session owner.
    \item $\mathsf{pid} \in \{1,\dots,n_P,\bot\}$: the identifier of the alleged peer ($\bot$ means unauthenticated).
    \item $\mathsf{status} \in \{\mathsf{active},\mathsf{reject},\mathsf{accept}\}$: the current status.
    \item $k$: the session key, either $\bot$ or a tuple of subkeys. In SSH, $k = (k_e,k_d)$ where $k_e$ and $k_d$ are bidirectional authenticated encryption keys.
    \item $\mathsf{sid}$: a session identifier defined by the protocol, used to match communicating sessions.
    \item $s_e$, $s_d$: state for the stateful authenticated encryption and decryption algorithms.
    \item Additional protocol-specific state (e.g., nonces exchanged, handshake message counters).
  \end{itemize}
  For the security experiment, each session also initialises a random bit $\pi_i^s.b \leftarrow \{0,1\}$ that is used in the channel security challenge.
\end{definition}

\begin{definition}[ACCE protocol]
  An ACCE protocol is a tuple of algorithms:
  \begin{itemize}
    \item $\mathsf{KeyGen}(1^\lambda) \to (sk,pk)$: outputs a long-term key pair.
    \item Handshake algorithms $\mathsf{AlgI}_\ell$, $\mathsf{AlgR}_\ell$ (for initiator and responder):
      on input $(sk,pk)$ and an incoming message $m$, they update the session state and return an outgoing message $m'$.
      Eventually they set $\pi.\mathsf{pid}$, $\pi.\mathsf{status}$, $\pi.k$, $\pi.\mathsf{sid}$.
    \item Stateful authenticated encryption $\mathsf{Enc}(\pi.k_e, m, \pi.s_e) \to (C, \pi.s_e')$.
    \item Stateful authenticated decryption $\mathsf{Dec}(\pi.k_d, C, \pi.s_d) \to (m', \pi.s_d')$.
  \end{itemize}
  All algorithms implicitly take global protocol parameters, including the list of all trusted peer public keys.
\end{definition}

\subsection{Adversary interaction}

The adversary $\mathcal{A}$ (or $\mathcal{Q}$ in the post-quantum setting) controls all network communication. It can direct parties to initiate sessions, deliver messages, reorder, alter, or drop packets. The adversary interacts with honest parties via the following oracles.

\begin{itemize}
\item $\mathsf{OSend}(i,s,m)$: The adversary sends message $m$ to session $\pi_i^s$.  
  If $\pi_i^s$ does not exist, it is created with $\rho = \mathsf{init}$ and $\mathsf{status} = \mathsf{active}$ (using the special symbol $m = \mathsf{init}$).  
  The session processes $m$, updates its state, and returns the outgoing protocol message (or $\bot$ if the session terminates).  
  This oracle models active adversarial control over the network.

\item $\mathsf{OExecute}(i,s,j,t)$: The challenger executes a full honest protocol run between two parties $P_i$ and $P_j$, creating sessions $\pi_i^s$ and $\pi_j^t$ with matching roles ($\rho = \mathsf{init}$ for $i$, $\rho = \mathsf{resp}$ for $j$).  
  All handshake messages are generated honestly; the adversary only receives the complete transcript of the session (the sequence of messages exchanged).  
  This oracle models \emph{passive} eavesdropping, in particular ``harvest now, decrypt later'' attacks.  
  Unlike $\mathsf{OSend}$, the adversary cannot modify or inject messages into the sessions created by $\mathsf{OExecute}$.  
  In the post-quantum ACCE model, $\mathsf{OSend}$ is replaced by $\mathsf{OExecute}$; in the classical model both oracles may be available.

\item $\mathsf{OReveal}(i,s)$: If $\pi_i^s$ has accepted (i.e., $\pi_i^s.\mathsf{status} = \mathsf{accept}$), the oracle returns the session key $\pi_i^s.k$ to the adversary.  
  This models leakage of session keys (e.g., via compromise of the host).

\item $\mathsf{OCorrupt}(i)$: The adversary obtains the long-term private key $sk_i$ of party $P_i$.  
  After this query, the party is considered \emph{corrupted}.  
  This models long-term key compromise.

\item $\mathsf{OEncrypt}(i,s,m)$: If session $\pi_i^s$ has accepted and is in the channel phase, the oracle returns a ciphertext $c \leftarrow \mathsf{Enc}(\pi_i^s.k_e, m, \pi_i^s.s_e)$ and updates the encryption state.  
  This is used in the channel security experiment to provide encryption queries.

\item $\mathsf{ODecrypt}(i,s,c)$: If session $\pi_i^s$ has accepted and is in the channel phase, the oracle returns the decryption $(m', \pi_i^s.s_d) \leftarrow \mathsf{Dec}(\pi_i^s.k_d, c, \pi_i^s.s_d)$ (or $\bot$ if the ciphertext is invalid).  
  This provides decryption queries.
\end{itemize}

The adversary may adaptively interleave these queries. The security experiments impose freshness conditions (e.g., no $\mathsf{OReveal}$ on the test session or its partner, no $\mathsf{OCorrupt}$ before acceptance, etc.) as defined in the channel security definitions.

\subsection{ACCE Security}

We now define the security properties for the SSH protocol with mutual authentication (server authenticated via signatures, client authenticated via KEM-based challenge-response). The definitions follow the ACCE framework, separating authentication and channel security.

\begin{definition}[Matching sessions]
  We say that session $\pi_j^t$ \textbf{matches} $\pi_i^s$ if:
  \begin{enumerate}
    \item $\pi_i^s.\rho \neq \pi_j^t.\rho$ (the roles are opposite);
    \item $\pi_i^s.\mathsf{c} = \pi_j^t.\mathsf{c}$ (the same ciphersuite is used);
    \item $\pi_i^s.\mathsf{sid}$ prefix-matches $\pi_j^t.\mathsf{sid}$, i.e., either:
      \begin{itemize}
        \item if $\pi_i^s$ sent the last message in $\pi_i^s.\mathsf{sid}$, then $\pi_j^t.\mathsf{sid}$ is a prefix of $\pi_i^s.\mathsf{sid}$; or
        \item if $\pi_j^t$ sent the last message in $\pi_i^s.\mathsf{sid}$, then $\pi_i^s.\mathsf{sid} = \pi_j^t.\mathsf{sid}$.
      \end{itemize}
  \end{enumerate}
  In our analysis we consider the standard SSH handshake (without abbreviated handshake), so matching sessions will have identical session identifiers: $\pi_i^s.\mathsf{sid} = \pi_j^t.\mathsf{sid}$.
\end{definition}

\begin{definition}[(Post-quantum) Authentication security]\label{pq auth security}
  Let $\pi_i^s$ be a session. We say that $\pi_i^s$ accepts maliciously if:
  \begin{itemize}
    \item $\pi_i^s.\mathsf{status} = \mathsf{accept}$;
    \item $\pi_i^s.\mathsf{pid} = j \neq \bot$ and no $\mathsf{OCorrupt}(j)$ query was issued before $\pi_i^s$ accepted;
    \item but there is no unique session $\pi_j^t$ that matches $\pi_i^s$.
  \end{itemize}
  For mutual authentication, both initiator (client) and responder (server) sessions can accept maliciously.
  Define $\mathrm{Adv}_{\mathsf{SSH}}^{\mathsf{acce}\mbox{-}\mathsf{mu\mbox{-}auth}}(\mathcal{A})$ as the probability that, when a PPT adversary $\mathcal{A}$ terminates in the ACCE experiment, there exists \emph{any} session (initiator or responder) that has accepted maliciously.

  The post-quantum mutual-authentication experiment is defined identically, except that the adversary is a quantum polynomial-time adversary $\mathcal{Q}$. We define $\mathrm{Adv}_{\mathsf{SSH}}^{\mathsf{pq\text{-}acce\text{-}mu\text{-}auth}}(\mathcal{Q})$ as the probability that, when $\mathcal{Q}$ terminates in the same ACCE authentication experiment, there exists \emph{any} session, initiator or responder, that has accepted maliciously.
\end{definition}

\begin{definition}[Channel security]\label{channel security}
  Suppose a PPT adversary $\mathcal{A}$ has access to $\mathsf{OSend}$, $\mathsf{OExecute}$, $\mathsf{OReveal}$, $\mathsf{OCorrupt}$, $\mathsf{OEncrypt}$, $\mathsf{ODecrypt}$. At the end of the experiment, $\mathcal{A}$ outputs $(i,s,b')$. We say that $\mathcal{A}$ answers the encryption challenge correctly if:
  \begin{itemize}
    \item $\pi_i^s.\mathsf{status} = \mathsf{accept}$;
    \item No $\mathsf{OCorrupt}(j)$ query was issued before $\pi_i^s$ accepted where $\pi_i^s.\mathsf{pid} = j$ and $\pi_i^s$ does \textbf{not} have a matching session;
    \item No $\mathsf{OReveal}(i,s)$ query was issued;
    \item For any $\pi_j^t$ that matches $\pi_i^s$, no $\mathsf{OReveal}(j,t)$ query was issued;
    \item $\pi_i^s.b = b'$.
  \end{itemize}
  Define $\mathrm{Adv}_{\mathsf{SSH}}^{\mathsf{acce\text{-}aenc}}(\mathcal{A}) = \left| p - \frac12 \right|$, where $p$ is the probability that $\mathcal{A}$ answers the encryption challenge correctly and either $\pi_i^s.\rho = \mathsf{init}$ or both $\pi_i^s.\rho = \mathsf{resp}$ and a matching session exists.
\end{definition}

In the post-quantum setting, the adversary is a quantum polynomial-time (QPT) algorithm and the $\mathsf{OSend}$ oracle is replaced by $\mathsf{OExecute}$. This models passive eavesdropping with future quantum decryption capability.

\begin{definition}[Post-quantum channel security]
  Let $\mathcal{Q}$ be a QPT adversary with access to $\mathsf{OExecute}$, $\mathsf{OReveal}$, $\mathsf{OCorrupt}$, $\mathsf{OEncrypt}$, $\mathsf{ODecrypt}$. The experiment is the same as in Definition~\ref{channel security}, except that $\mathcal{Q}$ cannot use $\mathsf{OSend}$ (all honest sessions are generated via $\mathsf{OExecute}$). We say $\mathcal{Q}$ answers the encryption challenge correctly under the same conditions. Define $\mathrm{Adv}_{\mathsf{SSH}}^{\mathsf{pq}\text{-}\mathsf{acce}\text{-}\mathsf{aenc}}(\mathcal{Q}) = \left| p - \frac12 \right|$ analogously.
\end{definition}

Because $\mathsf{OExecute}$ guarantees that every session has a matching partner (as it creates both sides honestly), authentication is automatically satisfied. Therefore post-quantum ACCE security reduces to channel security only.

\subsection{Mutual Auth. Mode Security}\label{mutual-auth-proof}

\begin{figure*}
\centering
\begin{pcvstack}[center,space=2em]  

    \begin{pchstack}[boxed,center,space=2em]

        \begin{pcvstack}[space=0.5em]
            \procedure[linenumbering=false, mode=text, headlinecmd={}, headlinesep=0pt, bodylinesep=0pt]{\textbf{\large Negotiation}}{}
            \procedure[linenumbering, mode=text]{1. init $\rightarrow$ resp: KEXINIT}{
                $r_C \sample \{0,1\}^\mu$ \\
                send $\mathsf{KEXINIT}\leftarrow (r_C, \overrightarrow{\mathsf{SP}}_C)$ \\
                $\pi.\rho \leftarrow \mathsf{init}$ \\
                $\pi.\mathsf{status} \leftarrow \mathsf{active}$
            }
            \procedure[linenumbering, mode=text]{2. resp $\rightarrow$ init: KEXREPLY}{
                $r_S \sample \{0,1\}^\mu$ \\
                send $\mathsf{KEXREPLY}\leftarrow (r_S, \overrightarrow{\mathsf{SP}}_S)$ \\
                $\pi.\rho \leftarrow \mathsf{resp}$ \\
                $\pi.\mathsf{status} \leftarrow \mathsf{active}$ \\
                $\pi.c \leftarrow \mathsf{neg}(\overrightarrow{\mathsf{SP}}_C, \overrightarrow{\mathsf{SP}}_S)$
            }
            \procedure[linenumbering, mode=text]{3. init}{
                $\pi.c \leftarrow \mathsf{neg}(\overrightarrow{\mathsf{SP}}_C, \overrightarrow{\mathsf{SP}}_S)$
            }
        \end{pcvstack}

        \begin{pcvstack}[space=0.5em]
            \procedure[linenumbering=false, mode=text, headlinecmd={}, headlinesep=0pt, bodylinesep=0pt]{\textbf{\large Signed KEX (for both modes)}}{}
            \procedure[linenumbering, mode=text]{4. init $\rightarrow$ resp: KEX\_KEM\_INIT}{
                $(pk, sk) \leftarrow \mathsf{KEM}_{\pi.c}.\mathsf{Gen}()$ \\
                send $\mathsf{KEX\_KEM\_INIT} \leftarrow pk$
            }
            \procedure[linenumbering, mode=text]{5. resp $\rightarrow$ init: KEX\_KEM\_REPLY and NEWKEYS}{
                $(K, ct) \leftarrow \mathsf{KEM}_{\pi.c}.\mathsf{Encap}(pk)$ \\
                $x \leftarrow V_C \parallel V_S \parallel \mathsf{KEXINIT} \parallel \mathsf{KEXREPLY} \parallel pk_{S,\pi.c} \parallel pk \parallel ct$ \\
                $(\pi.\mathsf{sid}, \pi.k) \leftarrow \mathsf{PRF}_{\mathsf{SSH}}(K, x)$ \\
                $\sigma_S \leftarrow \mathsf{DSS}_{\pi.c}.\mathsf{Sign}(sk_{S,\pi.c}, \pi.\mathsf{sid})$ \\
                send $\mathsf{KEX\_KEM\_REPLY} \leftarrow (ct, pk_{S,\pi.c}, \sigma_S)$ \\
                send $\mathsf{NEWKEYS}$
            }
            \procedure[linenumbering, mode=text]{6. init $\rightarrow$ resp: NEWKEYS}{
                $K \leftarrow \mathsf{KEM}_{\pi.c}.\mathsf{Decap}(sk, ct)$ \\
                $x \leftarrow V_C \parallel V_S \parallel \mathsf{KEXINIT} \parallel \mathsf{KEXREPLY} \parallel pk_{S,\pi.c} \parallel pk \parallel ct$ \\
                $(\pi.\mathsf{sid}, \pi.k) \leftarrow \mathsf{PRF}_{\mathsf{SSH}}(K, x)$ \\
                \textbf{if} $\mathsf{DSS}_{\pi.c}.\mathsf{Vrfy}(pk_{S,\pi.c}, \sigma_S, \pi.\mathsf{sid}) = \bot$ \textbf{then} \\
                \qquad $\pi.\mathsf{status} \leftarrow \mathsf{reject}$ and terminate \\
                $\pi.\mathsf{pid} \leftarrow S$, the owner of $pk_{S,\pi.c}$ \\
                send $\mathsf{NEWKEYS}$
            }
        \end{pcvstack}
    \end{pchstack}

    \begin{pchstack}[boxed,center,space=2em]
        \begin{pcvstack}[space=0.5em]
            \procedure[linenumbering=false, mode=text, headlinecmd={}, headlinesep=0pt, bodylinesep=0pt]{\textbf{\large Mutual Auth. Mode}}{}
            \procedure[linenumbering, mode=text]{7. init $\rightarrow$ resp: AUTH\_REQUEST}{
                $A \leftarrow username \parallel service \parallel \mathsf{publickey\text{-}kem} \parallel alg \parallel pk_{C,\pi.c}$ \\
                \qquad \textbf{/} $alg$ specifies $\mathsf{KEM}'_{\pi.c}$ \\
                send $\mathsf{AUTH\_REQUEST} \leftarrow A$
            }
            \procedure[linenumbering, mode=text]{8. resp $\rightarrow$ init: AUTH\_CHALLENGE or AUTH\_FAILURE}{
                \textbf{if} $\mathsf{Check}(username, service, \mathsf{publickey\text{-}kem})$ \textbf{then} \\
                \qquad $\pi.\mathsf{status} \leftarrow \mathsf{reject}$ and terminate \\
                \textbf{if} $\pi.\mathsf{status} = \mathsf{active}$ \textbf{then} \\
                \qquad $(K', ct') \leftarrow \mathsf{KEM}'_{\pi.c}.\mathsf{Encap}(pk_{C,\pi.c})$ \\
                \qquad $A' \leftarrow alg \parallel pk_{C,\pi.c} \parallel ct'$ \\
                \qquad send $\mathsf{AUTH\_CHALLENGE} \leftarrow A'$ \\
                \textbf{if} $\pi.\mathsf{status} = \mathsf{reject}$ \textbf{then} \\
                \qquad send $\mathsf{AUTH\_FAILURE}$ and terminate
            }
            \procedure[linenumbering, mode=text]{9. init $\rightarrow$ resp: AUTH\_RESPONSE}{
                $K' \leftarrow \mathsf{KEM}'_{\pi.c}.\mathsf{Decap}(sk_{C,\pi.c}, ct')$ \\
                $\mathsf{ctx} \leftarrow A \parallel \mathsf{AUTH\_CHALLENGE}$ \\
                $y \leftarrow \pi.\mathsf{sid} \parallel \mathsf{ctx}$ \\
                $ca \leftarrow \mathsf{PRF}_{\mathsf{CA}}(K', y)$\\
                send $\mathsf{AUTH\_RESPONSE} \leftarrow ca$
            }
        \end{pcvstack}
        \begin{pcvstack}[space=0.5em]
            \procedure[linenumbering, mode=text]{10. resp $\rightarrow$ init: AUTH\_SUCCESS or AUTH\_FAILURE}{
                $\mathsf{ctx}' \leftarrow \mathsf{AUTH\_REQUEST} \parallel A'$ \\
                $y' \leftarrow \pi.\mathsf{sid} \parallel \mathsf{ctx}'$ \\
                $ca' \leftarrow \mathsf{PRF}_{\mathsf{CA}}(K', y')$ \\
                \textbf{if} $ca' \neq ca$ \textbf{then} \\
                \qquad $\pi.\mathsf{status} \leftarrow \mathsf{reject}$ \\
                \textbf{if} $\pi.\mathsf{status} = \mathsf{active}$ \textbf{then} \\
                \qquad $\pi.\mathsf{status} \leftarrow \mathsf{accept}$ \\
                \textbf{if} $\pi.\mathsf{status} = \mathsf{accept}$ \textbf{then} \\
                \qquad send $\mathsf{AUTH\_SUCCESS}$ \\
                \textbf{if} $\pi.\mathsf{status} = \mathsf{reject}$ \textbf{then} \\
                \qquad send $\mathsf{AUTH\_FAILURE}$ and terminate
            }
            \procedure[linenumbering, mode=text]{11. init}{
                \textbf{if} $\mathsf{AUTH\_FAILURE}$ \textbf{then} \\
                \qquad $\pi.\mathsf{status} \leftarrow \mathsf{reject}$ and terminate \\
                \textbf{if} $\mathsf{AUTH\_SUCCESS}$ \textbf{then} \\
                \qquad $\pi.\mathsf{status} \leftarrow \mathsf{accept}$
            }
        \end{pcvstack}
    \end{pchstack}
\end{pcvstack}

\caption{Description of post-quantum SSH handshake protocol with KEM-based Mutual Auth. Mode. In steps 1 and 2, algorithm $\mathsf{neg}$ negotiates a ciphersuite to be used in the protocol. $V_C$ and $V_S$ are version strings. In steps 5 and 6, $\mathsf{PRF}_{\mathsf{SSH}}$ is defined in~\cite{pqacce}. In step 8, algorithm $\mathsf{Check}$ checks the accessibility of username for service via some authentication mechanism, $\mathsf{none}$ or $\mathsf{publickey\mbox{-}kem}$. In step 9, $\mathsf{PRF}_{\mathsf{CA}}$ is defined in Section~\ref{subsec:prfca}. $\mathsf{ctx}\leftarrow A||A'$ denotes the canonical $\mathsf{Encode}_{\mathsf{SSH}}(\mathsf{msg})$ defined in Section III, including the authentication request, the challenge containing $ct'$, and the $\mathsf{AUTH\_RESPONSE}$ message type. We note that from steps 7–11, their protocol is run within $\mathsf{BSAE}$ using $\pi.k$.}
\label{fig:ssh-full-protocol}
\end{figure*}

\begin{proof}
We define $\mathsf{break}_k$ as the event that a session accepts maliciously (Definition~\ref{pq auth security}) in Game $k$. The proof proceeds through a sequence of games, following the protocol flow and state definitions depicted in Figure~\ref{fig:ssh-full-protocol}.

\textbf{Game 0.}  
This is the original ACCE authentication experiment where the adversary $\mathcal{A}$ interacts with the real protocol. Hence
\[
\mathsf{Adv}_{\mathsf{SSH}}^{\mathsf{acce\mbox{-}mu\mbox{-}auth}}(\mathcal{A}) = \Pr[\mathsf{break}_0].
\]

\textbf{Game 1.}  
The challenger aborts if any nonces (from $\mathsf{KEXINIT}$ or $\mathsf{KEXREPLY}$) collide. There are at most $n_P n_S$ sessions, each sampling a nonce of length $\mu$. The probability of a collision is at most $(n_P n_S)^2 / 2^{\mu+1}$, Thus
\[
\Pr[\mathsf{break}_0] \le \Pr[\mathsf{break}_1] + \frac{(n_P n_S)^2}{2^{\mu}}.
\]

\textbf{Game 2.}  
The challenger maintains a list of all inputs and outputs of the hash function $\mathsf{H}$. If two distinct inputs produce the same output, the game aborts and we output the collision to a collision-resistance challenger. Therefore
\[
\Pr[\mathsf{break}_1] \le \Pr[\mathsf{break}_2] + \mathsf{Adv}_{\mathsf{H}}^{\mathsf{coll}}(\mathcal{B}_1),
\]
where $\mathcal{B}_1$ is a reduction against the collision resistance of $\mathsf{H}$.

\textbf{Game 3.}  
The challenger guesses the first session $\pi_i^s$ that will accept maliciously. There are $n_P$ parties and at most $n_S$ sessions per party, so the guess is correct with probability $1/(n_P n_S)$. Hence
\[
\Pr[\mathsf{break}_2] \le n_P n_S \cdot \Pr[\mathsf{break}_3].
\]
From now on we fix this session $\pi_i^s$. Its role (initiator or responder) determines which case we consider. Note that $\pi_i^s.\mathsf{pid}=j\neq \perp$ and no $\mathsf{OCorrupt}(j)$ query was issued before $\pi_i^s$ accepted.

\paragraph{\textbf{Case A: $\mathbf{\pi_i^s}$ is a client session (initiator)}}
In this case the client accepts but there is no matching server session. 
We focus on the event $\mathsf{break}_3 \land \text{client}$, i.e., the first malicious acceptance occurs at a client session. For brevity, we denote this event by $\mathsf{break}_{3^{\mathsf{C}}}$. The client verifies a server signature $\sigma_S$ on the session identifier $\pi.\mathsf{sid}$. Because nonce and hash collisions are excluded by Games~1 and~2, each $\mathsf{sid}$ is unique. We now proceed with a sequence of games internal to this case.

\textbf{Game A1.}  
We abort the simulation if the test session $\pi_i^s$ accepts after receiving a signature that was never output by a session with a matching session identifier. Since nonce and hash collisions have been excluded, all values to be signed are distinct. Hence any such abort corresponds to a signature forgery.

We construct a reduction $\mathcal{B}_2$ that simulates the protocol as in Game~3. $\mathcal{B}_2$ receives a public key $pk^*$ from an $\mathsf{EUF\text{-}CMA}$ signature challenger, guesses the server identity $j = \pi_i^s.\mathsf{pid}$ (cost factor $n_P$), and sets $P_j$'s public key to $pk^*$. For all other servers keys are generated normally. Whenever $\mathcal{B}_2$ needs to sign a message on behalf of $P_j$, it queries the $\mathsf{EUF\text{-}CMA}$ signing oracle. If the test session $\pi_i^s$ accepts maliciously, then the pair $(\pi.\mathsf{sid}, \sigma_S)$ is a valid forgery, which $\mathcal{B}_2$ outputs to its challenger. Therefore
\[
\Pr[\mathsf{break}_{3^{\mathsf{C}}}] \le \Pr[\mathsf{break}_{\text{A1}}] + n_P \cdot \mathsf{Adv}_{\mathsf{DSS}}^{\mathsf{euf\text{-}cma}}(\mathcal{B}_2).
\]

From this point onward, the client's acceptance implies that the ephemeral KEM ciphertext $ct$ is authentic because it is bound by the server's signature.

\textbf{Game A2.}  
In this game, we replace the transport-layer KEM secret $K$ used by the target client session $\pi_i^s$ and its server partner $\pi_j^t$ with a uniformly random value $K^*$. We construct a reduction $\mathcal{B}_3$ that interacts with $\mathcal{A}$ and embeds an $\mathsf{IND\text{-}CPA}$ challenge into the transcript of these sessions.

$\mathcal{B}_3$ participates in the $\mathsf{IND\text{-}CPA}$ experiment for $\mathsf{KEM_e}$ and receives the challenge public key $pk^*$ together with a challenge ciphertext $ct^*$ and challenge key $K^*$. When the target client session $\pi_i^s$ would generate its ephemeral $\mathsf{KEM_e}$ key pair, $\mathcal{B}_3$ uses $pk^*$ as the client's public key; no corresponding secret key is required in the simulation. When the partner server session $\pi_j^t$ would encapsulate to this key, $\mathcal{B}_3$ supplies $(K^*,ct^*)$ as the encapsulation output. At the corresponding client decapsulation step, $\mathcal{B}_3$ directly supplies the same value $K^*$. Because malicious acceptance has already been excluded, the target sessions communicate without adversarial modification, and the challenge ciphertext received by $\pi_i^s$ is exactly $ct^*$.

If the challenge bit $b$ is $0$, then $(ct^*,K^*)$ is distributed as an honest encapsulation under $pk^*$, and the simulation is identical to Game~A1. If $b=1$, then $K^*$ is uniformly random and independent of $ct^*$, and the simulation is identical to Game~A2. Hence
\[
\Pr[\mathsf{break}_{\text{A1}}] \le \Pr[\mathsf{break}_{\text{A2}}] + \mathsf{Adv}_{\mathsf{KEM_e}}^{\mathsf{ind\text{-}cpa}}(\mathcal{B}_3).
\]

\textbf{Game A3.}  
We replace the values $H,k_1,\dots,k_6$ computed by $\pi_i^s$ and $\pi_j^t$ as $\mathsf{PRF}_{\mathsf{SSH}}(K^*, \mathsf{sid})$ with independent uniformly random strings $H^*, k_1^*, \dots, k_6^*$. Let $S = H \| k_1 \| \dots \| k_6$ and $S^* = H^* \| k_1^* \| \dots \| k_6^*$. A reduction $\mathcal{B}_4$ plays the PRF security game against $\mathsf{PRF}_{\mathsf{SSH}}$: it simulates the protocol as in Game~A2, but for the test session and its partner it forwards the $\mathsf{PRF}$ computation to its own challenger. If the challenger returns the real PRF output we are in Game~A2; if it returns a random string we are in Game~A3. Thus
\[
\Pr[\mathsf{break}_{\text{A2}}] \le \Pr[\mathsf{break}_{\text{A3}}] + \mathsf{Adv}_{\mathsf{PRF}_{\mathsf{SSH}}}^{\mathsf{prf}}(\mathcal{B}_4).
\]

\textbf{Game A4.}  
Now the keys $k_1^*, \dots, k_6^*$ are uniformly random and independent of the handshake. We replace the stateful authenticated encryption used in the channel phase by a $\mathsf{BSAE}$ challenger. A reduction $\mathcal{B}_5$ initialises a $\mathsf{BSAE}$ challenger and forwards all $\mathsf{Enc}$ and $\mathsf{Dec}$ operations of the test session $\pi_i^s$ and its matching partner $\pi_j^t$ to this challenger. The challenger uses either real encryption/decryption (if its internal bit is $0$) or returns random ciphertexts (if the bit is $1$). This replacement is sound because the keys are already random.

If $\pi_i^s$ accepts maliciously after this game, the adversary must have forged a valid ciphertext under the $\mathsf{BSAE}$ scheme, which would break $\mathsf{BSAE}$ security. Hence
\[
\Pr[\mathsf{break}_{\text{A3}}] \le \Pr[\mathsf{break}_{\text{A4}}] + \mathsf{Adv}_{\mathsf{BSAE}}^{\mathsf{bsae}}(\mathcal{B}_5).
\]

In Game~A4, all signatures are generated by legitimate parties, all session identifiers are unique, and the channel keys are independent random strings. The client $\pi_i^s$ can only accept if it receives a valid server signature (already excluded by Game~A1) and if the channel messages are authentic under the $\mathsf{BSAE}$ scheme. The $\mathsf{BSAE}$ challenger guarantees that no adversary can forge a ciphertext, so $\mathsf{break}_{\text{A4}}$ occurs with probability $0$. Therefore
\[
\Pr[\mathsf{break}_{\text{A4}}] = 0.
\]

Collecting the inequalities for Case~A, we obtain
\begin{align*}
\Pr[\mathsf{break}_3 \land \text{client}] 
&\le n_P \Bigl( \mathsf{Adv}_{\mathsf{DSS}}^{\mathsf{euf\text{-}cma}}(\mathcal{B}_2) \\
&\qquad + \mathsf{Adv}_{\mathsf{KEM_e}}^{\mathsf{ind\text{-}cpa}}(\mathcal{B}_3) \\
&\qquad + \mathsf{Adv}_{\mathsf{PRF}_{\mathsf{SSH}}}^{\mathsf{prf}}(\mathcal{B}_4) \\
&\qquad + \mathsf{Adv}_{\mathsf{BSAE}}^{\mathsf{bsae}}(\mathcal{B}_5) \Bigr).
\end{align*}

\paragraph{\textbf{Case B: $\mathbf{\pi_i^s}$ is a server session (responder)}}
Now the server accepts but there is no matching client session. We consider the event $\mathsf{break}_3 \land \mathsf{server}$, denoted $\mathsf{break}_{3^{\mathsf{s}}}$. In our mutual-authentication protocol, the server verifies a client authentication value $ca = \mathsf{PRF}_{\mathsf{CA}}(K_C, y)$, where $(K_C,ct_C)$ is generated by encapsulating to the client's authorized long-term KEM public key and $y=\mathsf{sid}\parallel\mathsf{ctx}$ is the session-bound PRF input.

\textbf{Game B1.}
We guess the client identity $j = \pi_i^s.\mathsf{pid}$; this is correct with probability $1/n_P$. We replace the static $\mathsf{KEM}_{\mathsf{C}}$ public key of client $P_j$ with a challenge public key $pk^*$ obtained from an $\mathsf{IND\text{-}CCA}$ challenger. In the target authentication attempt, the server's encapsulation to $pk^*$ is simulated using the challenge ciphertext and challenge key. For other honest sessions involving $P_j$, any required decapsulation of non-challenge ciphertexts is answered using the $\mathsf{IND\text{-}CCA}$ decapsulation oracle. If $\mathcal{A}$ distinguishes this game from the previous one, we obtain an adversary against the $\mathsf{IND\text{-}CCA}$ security of $\mathsf{KEM}_{\mathsf{C}}$. Hence
\[
\Pr[\mathsf{break}_{3^{\mathsf{s}}}]
\le
n_P \Bigl(
  \Pr[\mathsf{break}_{\text{B1}}]
  +
  \mathsf{Adv}_{\mathsf{KEM}_{\mathsf{C}}}^{\mathsf{ind\text{-}cca}}(\mathcal{B}_6)
\Bigr).
\]

\textbf{Game B2.}
In this game, the server's expected authentication value $ca'=\mathsf{PRF}_{\mathsf{CA}}(K_C,y)$ for the target authentication attempt is replaced by a uniformly random string of the same length. Since the target key $K_C$ is hidden from the adversary after the $\mathsf{IND\text{-}CCA}$ replacement, this hop reduces to the PRF security of $\mathsf{PRF}_{\mathsf{CA}}$. A reduction $\mathcal{B}_7$ simulates the game and forwards the relevant query to a PRF challenger. Thus
\[
\Pr[\mathsf{break}_{\text{B1}}]
\le
\Pr[\mathsf{break}_{\text{B2}}]
+
\mathsf{Adv}_{\mathsf{PRF}_{\mathsf{CA}}}^{\mathsf{prf}}(\mathcal{B}_7).
\]

\textbf{Game B3.}
Let
\[
y^\star=\pi^\star.\mathsf{sid}\parallel\mathsf{ctx}^\star
\]
be the input to $\mathsf{PRF}_{\mathsf{CA}}$ in the target authentication attempt.  Here $\mathsf{ctx}^\star\leftarrow A^\star||A^{\prime\star}$ is the common authentication context constructed for the target attempt. It abstracts the canonical encoding $\mathsf{Encode}_{\mathsf{SSH}}(\mathsf{msg}^\star)$ defined in Section~III, including the authentication request and challenge ciphertext, while $\pi^\star.\mathsf{sid}$ is established by the corresponding transport session. Although the reduction may need to evaluate $\mathsf{PRF}_{\mathsf{CA}}$ when simulating honest non-target authentication responses, we condition on the event that $y^\star$ is not evaluated before $\mathcal{A}$ sends its response $ca$. In the absence of a matching client session, such a query would require another simulated authentication attempt to use the same authentication context and the same session identifier. The former fixes the target challenge and method transcript, while the latter is derived from the transport-layer exchange and contains the randomness already protected by the nonce-collision and exchange-hash-collision exclusions in the previous games. Thus, the probability that $y^\star$ is hit through a non-target simulation is negligible and is absorbed into the preceding bad-event bounds.

The value $ca'$ is therefore a uniformly random $\ell$-bit string and is independent of the adversary's view before it sends $ca$. Hence,
\[
  \Pr[\mathsf{Guess}] \leq 2^{-\ell}.
\]
Since this game only records an event and does not change the simulation,
\[
\Pr[\mathsf{break}_{\text{B2}}]
=
\Pr[\mathsf{break}_{\text{B3}}].
\]

\textbf{Game B4.}
This game is identical to Game~B3, except that the target server session always rejects the received authentication response. Games B3 and B4 differ only if $\mathsf{Guess}$ occurs. Therefore,
\[
\Pr[\mathsf{break}_{\text{B3}}]
\le
\Pr[\mathsf{break}_{\text{B4}}] + 2^{-\ell}.
\]
In Game~B4, the target server session cannot accept maliciously, and hence
\[
\Pr[\mathsf{break}_{\text{B4}}] = 0.
\]

Combining the steps for Case~B, we obtain
\begin{align*}
\Pr[\mathsf{break}_3 \land \text{server}]
&\le n_P \Bigl( \mathsf{Adv}_{\mathsf{KEM}_{\mathsf{C}}}^{\mathsf{ind\text{-}cca}}(\mathcal{B}_6) \\
&\qquad + \mathsf{Adv}_{\mathsf{PRF}_{\mathsf{CA}}}^{\mathsf{prf}}(\mathcal{B}_7) + 2^{-\ell} \Bigr).
\end{align*}

\paragraph{\textbf{Result analysis}}  
Combining the two cases with the losses from Games~1–3, we have
\begin{align*}
    \mathsf{Adv}_{\mathsf{SSH}}^{\mathsf{acce\text{-}mu\text{-}auth}}(\mathcal{A})
    &\le \frac{(n_P n_S)^2}{2^{\mu}} + \mathsf{Adv}_{\mathsf{H}}^{\mathsf{coll}}(\mathcal{B}_1) \\
    \quad &+ n_P^2 n_S \Bigl( 
            \mathsf{Adv}_{\mathsf{DSS}}^{\mathsf{euf\text{-}cma}}(\mathcal{B}_2)
            + \mathsf{Adv}_{\mathsf{KEM}_{\mathsf{e}}}^{\mathsf{ind\text{-}cpa}}(\mathcal{B}_3) \\
    \quad &+ \mathsf{Adv}_{\mathsf{PRF}_{\mathsf{SSH}}}^{\mathsf{prf}}(\mathcal{B}_4)
            + \mathsf{Adv}_{\mathsf{BSAE}}^{\mathsf{bsae}}(\mathcal{B}_5) \\
    \quad &+ \mathsf{Adv}_{\mathsf{KEM}_{\mathsf{C}}}^{\mathsf{ind\text{-}cca}}(\mathcal{B}_6)
            + \mathsf{Adv}_{\mathsf{PRF_{CA}}}^{\mathsf{prf}}(\mathcal{B}_7) + 2^{-\ell} \Bigr).
\end{align*}
This completes the proof.
\end{proof}

\section{Concurrent Server Evaluation}
\label{app:concurrent-eval}

We compare KEMUAuth with ML-DSA-65 as the client-authentication method while fixing the transport key exchange to \texttt{mlkem768x25519-sha256} and the server host-key algorithm to \texttt{ssh-mldsa-65}. The tests run over loopback without emulated delay or packet loss. The complete \texttt{sshd} process tree is pinned to two dedicated CPU cores and placed in a cgroup v2, while the client workers use a disjoint CPU set. Connection reuse and non-tested authentication methods are disabled, and each worker repeatedly establishes a fresh SSH connection, completes authentication, executes \texttt{true}, and closes the connection. For each worker concurrency $N\in\{1,8,16,32,64\}$, we perform five runs with a 5-second warm-up, a 30-second measurement window, and a 10-second cooldown. Throughput and latency are derived from completed connections, while CPU time and aggregate peak memory are collected from \texttt{cpu.stat} and \texttt{memory.peak}. We report the median result across the five runs.

\begin{table}[htbp]
\centering
\caption{Concurrent SSH throughput and server CPU cost}
\label{tab:concurrent-throughput}
\scriptsize
\setlength{\tabcolsep}{3.5pt}
\renewcommand{\arraystretch}{1.3}
\providecommand{\thickhline}{\noalign{\hrule height 0.8pt}}
\providecommand{\spacedhline}{\noalign{\vskip 1.5pt}\hline\noalign{\vskip 1.5pt}}

\begin{tabular}{c l| r r r}
\thickhline
\noalign{\vskip 2pt}
$N$ & Authentication & Successful & Server CPU & CPU per \\
& Method & Throughput & Utilization & Connection \\
& & (conn/s) & (\%) & (ms) \\
\spacedhline

1
& KEMUAuth  & 11.36 & 11.2 & 19.67 \\
1
& ML-DSA-65 & 10.87 & 10.7 & 19.70 \\
\spacedhline

8
& KEMUAuth  & 84.48 & 74.4 & 17.57 \\
8
& ML-DSA-65 & 81.75 & 72.5 & 17.72 \\
\spacedhline

16
& KEMUAuth  & 107.99 & 97.6 & 18.06 \\
16
& ML-DSA-65 & 106.46 & 97.0 & 18.22 \\
\spacedhline

32
& KEMUAuth  & 107.65 & 99.5 & 18.49 \\
32
& ML-DSA-65 & 106.48 & 99.5 & 18.69 \\
\spacedhline

64
& KEMUAuth  & 106.27 & 99.7 & 18.97 \\
64
& ML-DSA-65 & 105.90 & 99.8 & 19.06 \\

\noalign{\vskip 2pt}
\thickhline
\end{tabular}%

\end{table}

Table~\ref{tab:concurrent-throughput} shows that KEMUAuth achieves slightly higher successful throughput and lower CPU time per completed connection across the tested concurrency levels. At $N=1$ and $N=8$, its connection rate is 4.5\% and 3.3\% higher than that of ML-DSA-65, respectively. This represents a modest advantage while spare server capacity remains available. As the CPU approaches saturation, the throughput difference decreases to 1.4\%, 1.1\%, and 0.35\% at $N=16$, $N=32$, and $N=64$, respectively. The two methods therefore provide closely matched saturated throughput. This narrowing is expected because a complete SSH connection also includes transport key exchange, server authentication, process management, packet processing, and scheduling, which dilute the primitive-level difference. The limited change in throughput from $N=32$ to $N=64$ indicates that the server has reached its processing capacity, so additional concurrency primarily increases queueing rather than the connection completion rate. KEMUAuth nevertheless retains slightly lower CPU time per completed connection throughout the experiment, while the two methods remain closely matched under saturated load.

\begin{table}[htbp]
\centering
\caption{Tail latency and aggregate memory under concurrent load}
\label{tab:concurrent-latency-memory}
\scriptsize
\setlength{\tabcolsep}{5.0pt}
\renewcommand{\arraystretch}{1.3}
\providecommand{\thickhline}{\noalign{\hrule height 0.8pt}}
\providecommand{\spacedhline}{\noalign{\vskip 1.5pt}\hline\noalign{\vskip 1.5pt}}

\begin{tabular}{c l| r r}
\thickhline
\noalign{\vskip 2pt}
$N$ & Authentication & $95^{th}$ Latency & Peak Memory \\
& Method & (ms) & (MB) \\
\spacedhline

1
& KEMUAuth  & 79  & 9.3 \\
1
& ML-DSA-65 & 83  & 9.3 \\
\spacedhline

8
& KEMUAuth  & 94  & 39.6 \\
8
& ML-DSA-65 & 97  & 40.1 \\
\spacedhline

16
& KEMUAuth  & 158 & 74.4 \\
16
& ML-DSA-65 & 161 & 74.6 \\
\spacedhline

32
& KEMUAuth  & 317 & 133.3 \\
32
& ML-DSA-65 & 321 & 133.5 \\
\spacedhline

64
& KEMUAuth  & 665 & 242.9 \\
64
& ML-DSA-65 & 680 & 244.4 \\

\noalign{\vskip 2pt}
\thickhline
\end{tabular}%

\end{table}

Table~\ref{tab:concurrent-latency-memory} shows that the $95^{th}$-percentile latency increases as connections queue behind the saturated CPU, rising from 79--83\,ms at $N=1$ to 665--680\,ms at $N=64$. KEMUAuth retains a slightly lower tail latency at every tested concurrency level, while both methods follow nearly identical scaling trends. Aggregate peak memory also remains closely matched and reaches 242.9\,MB for KEMUAuth and 244.4\,MB for ML-DSA-65 at $N=64$. Since this memory is dominated by the OpenSSH processes and common per-connection state, the additional KEM challenge state produces no visible aggregate-memory penalty in this experiment. Its protocol-specific size and lifecycle are evaluated separately in the pending-state experiment.

\begin{table}[htbp]
\centering
\caption{Estimated memory scaling under pending authentication state}
\label{tab:pending-state-memory}
\scriptsize
\setlength{\tabcolsep}{2.6pt}
\renewcommand{\arraystretch}{1.3}
\providecommand{\thickhline}{\noalign{\hrule height 0.8pt}}
\providecommand{\spacedhline}{\noalign{\vskip 1.5pt}\hline\noalign{\vskip 1.5pt}}

\begin{tabular}{c| r r r| r}
\thickhline
\noalign{\vskip 2pt}
Pending & \multicolumn{2}{c}{Estimated Peak Memory} &
Memory & Explicit KEM \\
Connections & KEMUAuth & ML-DSA-65 & Difference &
State \\
$P$ & (MiB) & (MiB) & (MiB) & (MiB) \\
\spacedhline

16  & 74.28  & 70.06  & 4.22 & 0.036 \\
32  & 124.03 & 119.86 & 4.17 & 0.073 \\
64  & 223.51 & 219.45 & 4.05 & 0.146 \\
128 & 422.48 & 418.65 & 3.83 & 0.291 \\
172 & 559.27 & 555.59 & 3.67 & 0.391 \\

\noalign{\vskip 2pt}
\thickhline
\end{tabular}%

\end{table}

To isolate the memory introduced by pending KEM challenges, we instrument challenge creation and removal and delay the client response so that multiple authentication attempts remain simultaneously pending. The ML-DSA-65 baseline is paused at the corresponding public-key authentication stage. Because independent runs occasionally reach slightly different maximum paused counts, Table~\ref{tab:pending-state-memory} uses separate linear fits to normalize both methods to the same pending count $P$, rather than directly comparing unmatched raw memory peaks. The fitted per-connection growth is approximately 3.11\,MiB for both methods, showing that aggregate memory is dominated by the ordinary OpenSSH process and protocol state. KEMUAuth has an estimated aggregate offset of about 4\,MiB, while its explicitly retained challenge context is only 2384 bytes per pending connection and remains below 0.4\,MiB at the largest observed count of 172. Pending entries return to zero after authentication completion or connection termination. These results show that KEMUAuth adds a small and bounded protocol-specific state without materially changing concurrent memory scaling.

\section{Detailed Evaluation Results}
\label{app:raw-eval}

This appendix provides the numerical results underlying Section~\ref{sec:evaluation}, following the setup in
Section~\ref{subsec:exp-setup}. Tables~\ref{tab:rtt-sensitivity} and~\ref{tab:loss-sensitivity} report sensitivity to configured RTT and random packet loss, respectively. Table~\ref{tab:client-auth-raw} gives the raw data for Figure~\ref{fig:client-auth-latency}; and Tables~\ref{tab:rtt-raw} and \ref{tab:tcp-iw-raw} provide the values underlying Figures~\ref{fig:rtt-latency} and~\ref{fig:initcwnd}, respectively.


\providecommand{\thickhline}{\noalign{\hrule height 0.8pt}}
\providecommand{\spacedhline}{%
  \noalign{\vskip 1.5pt}%
  \hline
  \noalign{\vskip 1.5pt}%
}

\begin{table*}[!t]
\centering
\scriptsize

\begin{minipage}[t]{0.48\textwidth}
\vspace{0pt}
\centering

\caption{End-to-end SSH handshake latency across configured RTTs}
\label{tab:rtt-sensitivity}

\begingroup
\setlength{\tabcolsep}{2.8pt}
\renewcommand{\arraystretch}{1.2}

\begin{tabular}{c|rr|rr|r}
\thickhline
\noalign{\vskip 2pt}

RTT &
\multicolumn{2}{c|}{KEMUAuth} &
\multicolumn{2}{c|}{ML-DSA-65} &
$\Delta_{50}$ \\

(ms) &
$50^{th}$ &
$95^{th}$ &
$50^{th}$ &
$95^{th}$ &
(ms) \\

\spacedhline

0   & 74.23   & 76.59   & 77.57   & 79.15   & 3.34 \\
20  & 329.97  & 331.81  & 331.51  & 333.88  & 1.54 \\
40  & 550.86  & 553.20  & 552.14  & 555.19  & 1.28 \\
60  & 770.96  & 772.55  & 772.16  & 773.85  & 1.20 \\
80  & 990.83  & 993.33  & 992.01  & 994.50  & 1.18 \\
100 & 1211.03 & 1212.57 & 1211.67 & 1214.08 & 0.64 \\
120 & 1431.31 & 1433.31 & 1432.28 & 1434.27 & 0.97 \\
160 & 1870.96 & 1872.57 & 1872.24 & 1874.26 & 1.28 \\
200 & 2310.94 & 2313.86 & 2312.00 & 2314.49 & 1.06 \\

\noalign{\vskip 2pt}
\thickhline
\end{tabular}

\endgroup

\vspace{1.0em}

\caption{End-to-end SSH handshake latency under random packet loss}
\label{tab:loss-sensitivity}

\begingroup
\setlength{\tabcolsep}{3.2pt}
\renewcommand{\arraystretch}{1.2}

\begin{tabular}{cl|rrr|r}
\thickhline
\noalign{\vskip 2pt}

Loss &
Authentication &
$50^{th}$ &
$95^{th}$ &
$99^{th}$ &
TCP Retx. \\

(\%) &
Method &
(ms) &
(ms) &
(ms) &
per Conn. \\

\spacedhline

0   & KEMUAuth  & 840.69 & 842.58  & 843.60  & 0.00 \\
0   & ML-DSA-65 & 841.68 & 843.63  & 844.63  & 0.00 \\

\spacedhline

0.1 & KEMUAuth  & 840.51 & 842.81  & 1115.33 & 0.07 \\
0.1 & ML-DSA-65 & 841.49 & 844.05  & 1119.90 & 0.09 \\

\spacedhline

0.5 & KEMUAuth  & 840.74 & 1116.64 & 1386.77 & 0.35 \\
0.5 & ML-DSA-65 & 841.71 & 1116.99 & 1844.37 & 0.35 \\

\spacedhline

1.0 & KEMUAuth  & 841.00 & 1184.91 & 1865.98 & 0.75 \\
1.0 & ML-DSA-65 & 841.93 & 1123.85 & 1876.07 & 0.78 \\

\spacedhline

2.0 & KEMUAuth  & 841.83 & 1774.46 & 2170.76 & 1.52 \\
2.0 & ML-DSA-65 & 842.77 & 1676.36 & 2140.88 & 1.52 \\

\noalign{\vskip 2pt}
\thickhline
\end{tabular}

\endgroup

\end{minipage}
\hfill
\begin{minipage}[t]{0.48\textwidth}
\vspace{0pt}
\centering

\caption{Raw end-to-end SSH handshake latency with different client-authentication algorithms}
\label{tab:client-auth-raw}

\begingroup
\setlength{\tabcolsep}{4pt}
\renewcommand{\arraystretch}{1.2}

\begin{tabular}{lc|r|rrr}
\thickhline
\noalign{\vskip 2pt}

Client &
Notation &
NIST PQ &
Mean &
$50^{th}$ &
$95^{th}$ \\

Algorithm &
&
Category &
(ms) &
(ms) &
(ms) \\

\spacedhline

Ed25519 &
ed25519 &
$\approx$0 bits &
842.042 &
841.962 &
845.765 \\

\spacedhline

ML-DSA-44 &
md44 &
Level 2 &
841.720 &
841.622 &
844.368 \\

ML-DSA-65 &
md65 &
Level 3 &
841.792 &
841.692 &
845.367 \\

ML-DSA-87 &
md87 &
Level 5 &
841.912 &
841.791 &
845.699 \\

\spacedhline

Falcon-512 &
falcon512 &
Level 1 &
841.746 &
841.656 &
845.233 \\

Falcon-1024 &
falcon1024 &
Level 5 &
841.940 &
841.850 &
846.634 \\

\spacedhline

SLH-DSA-SHA2-128f &
sd128f &
Level 1 &
920.703 &
920.556 &
928.641 \\

SLH-DSA-SHA2-192f &
sd192f &
Level 3 &
929.522 &
929.319 &
934.724 \\

SLH-DSA-SHA2-256f &
sd256f &
Level 5 &
1018.098 &
1017.955 &
1025.191 \\

\spacedhline

ML-KEM-512 &
mk512 &
Level 1 &
840.767 &
840.666 &
843.290 \\

ML-KEM-768 &
mk768 &
Level 3 &
840.779 &
840.681 &
843.480 \\

ML-KEM-1024 &
mk1024 &
Level 5 &
840.805 &
840.694 &
843.578 \\

\spacedhline

Password (yescrypt) &
password &
-- &
785.648 &
785.446 &
790.494 \\

\noalign{\vskip 2pt}
\thickhline
\end{tabular}

\endgroup

\end{minipage}

\end{table*}

\begin{table*}[!t]
\centering
\scriptsize

\caption{Raw end-to-end SSH handshake latency under different RTT regimes}
\label{tab:rtt-raw}

\begingroup
\setlength{\tabcolsep}{3.2pt}
\renewcommand{\arraystretch}{1.25}

\begin{tabular}{l l| c| r r r| r r r| r r r}
\thickhline
\noalign{\vskip 2pt}

Client &
Server &
Notation &
\multicolumn{3}{c|}{Close RTT = 37 ms} &
\multicolumn{3}{c|}{Intermediate RTT = 67 ms} &
\multicolumn{3}{c}{Long RTT = 163 ms} \\

Algorithm &
Algorithm &
&
Mean &
$50^{th}$ &
$95^{th}$ &
Mean &
$50^{th}$ &
$95^{th}$ &
Mean &
$50^{th}$ &
$95^{th}$ \\

&
&
&
(ms) &
(ms) &
(ms) &
(ms) &
(ms) &
(ms) &
(ms) &
(ms) &
(ms) \\

\spacedhline

ed25519 &
ed25519 &
ed25519 &
522.01 &
522.04 &
525.07 &
853.14 &
852.54 &
858.96 &
1911.21 &
1912.88 &
1918.35 \\

\spacedhline

md65 + ed25519 &
md65 + ed25519 &
h-md65 &
524.53 &
524.67 &
527.08 &
855.31 &
854.60 &
861.19 &
1913.69 &
1915.50 &
1920.41 \\

sd192f + ed25519 &
sd192f + ed25519 &
h-sd192f &
644.52 &
645.29 &
649.93 &
1033.83 &
1033.59 &
1042.34 &
2283.98 &
2285.79 &
2292.09 \\

\spacedhline

mk768 + ed25519 &
md65 + ed25519 &
h-mk768-md65 &
523.49 &
523.68 &
526.06 &
852.82 &
852.37 &
858.57 &
1913.37 &
1915.46 &
1920.11 \\

mk768 + ed25519 &
sd192f + ed25519 &
h-mk768-sd192f &
581.58 &
582.00 &
585.53 &
946.05 &
944.51 &
953.30 &
2096.68 &
2098.43 &
2103.70 \\

\spacedhline

ed25519 $\rightarrow$ md65 &
md65 + ed25519 &
ms-md65 &
603.55 &
603.85 &
606.64 &
997.94 &
992.99 &
1000.43 &
2247.47 &
2249.28 &
2254.37 \\

ed25519 $\rightarrow$ mk768 &
md65 + ed25519 &
ms-mk768 &
600.16 &
600.32 &
603.14 &
990.03 &
988.66 &
996.58 &
2244.87 &
2246.80 &
2251.31 \\

\spacedhline

mk768 &
md65 &
mk768-md65 &
518.80 &
518.84 &
521.24 &
849.51 &
848.75 &
854.14 &
1910.46 &
1911.72 &
1916.49 \\

\noalign{\vskip 2pt}
\thickhline
\end{tabular}

\endgroup

\vspace{1.8ex}

\caption{Raw SSH handshake latency under different TCP initial congestion-window sizes}
\label{tab:tcp-iw-raw}

\begingroup
\setlength{\tabcolsep}{3pt}
\renewcommand{\arraystretch}{1.3}

\begin{tabular}{c| r r r | r r r | r r r | r r r | r r r}
\thickhline
\noalign{\vskip 2pt}

\multirow{2}{*}{\begin{tabular}{c}
initcwnd\\
(MSS)
\end{tabular}} &
\multicolumn{3}{c|}{A: ed25519 + ed25519} &
\multicolumn{3}{c|}{B: md65 + md65} &
\multicolumn{3}{c|}{C: sd192f + sd192f} &
\multicolumn{3}{c|}{D: mk768 + md65} &
\multicolumn{3}{c}{E: mk768 + sd192f} \\

&
5th &
50th &
95th &
5th &
50th &
95th &
5th &
50th &
95th &
5th &
50th &
95th &
5th &
50th &
95th \\

\spacedhline

3  & 861.92 & 863.91 & 870.43 & 995.07 & 997.38 & 1002.71 & 1236.59 & 1238.92 & 1246.12 & 925.17 & 927.34 & 932.20 & 1014.53 & 1016.93 & 1021.97 \\
5  & 862.84 & 867.87 & 871.95 & 924.08 & 925.85 & 931.75 & 1167.80 & 1172.74 & 1178.35 & 856.51 & 858.14 & 863.13 & 1012.55 & 1014.71 & 1021.00 \\
7  & 862.07 & 867.14 & 871.34 & 855.44 & 858.63 & 864.56 & 1169.28 & 1173.71 & 1179.30 & 856.67 & 861.01 & 865.06 & 1012.38 & 1017.61 & 1021.15 \\
10 & 859.23 & 864.15 & 868.61 & 860.18 & 864.78 & 867.94 & 1036.09 & 1042.69 & 1047.26 & 859.44 & 863.44 & 867.21 & 945.45 & 950.18 & 955.06 \\
15 & 862.22 & 866.97 & 871.07 & 858.09 & 862.08 & 865.26 & 1038.97 & 1044.33 & 1048.63 & 856.95 & 859.97 & 863.45 & 946.43 & 951.21 & 955.22 \\
20 & 860.24 & 864.87 & 868.92 & 856.57 & 861.04 & 864.85 & 1036.98 & 1043.01 & 1047.72 & 857.35 & 861.69 & 865.77 & 944.99 & 950.44 & 954.85 \\
25 & 859.52 & 864.93 & 869.72 & 857.50 & 862.15 & 866.56 & 968.77 & 975.13 & 980.04 & 859.12 & 864.28 & 867.88 & 948.44 & 953.63 & 958.28 \\
30 & 861.46 & 864.13 & 870.20 & 858.35 & 860.82 & 866.06 & 904.29 & 907.19 & 912.07 & 859.54 & 862.40 & 866.26 & 879.94 & 883.03 & 887.58 \\
35 & 859.94 & 863.49 & 869.28 & 859.60 & 862.38 & 867.48 & 904.52 & 908.02 & 912.84 & 858.80 & 861.37 & 866.01 & 881.66 & 884.58 & 888.73 \\
40 & 859.60 & 861.44 & 867.76 & 857.23 & 859.22 & 864.47 & 901.47 & 904.88 & 910.10 & 856.06 & 858.28 & 863.77 & 879.60 & 882.63 & 888.24 \\
50 & 861.10 & 864.77 & 870.32 & 860.71 & 862.17 & 867.35 & 904.15 & 906.76 & 911.24 & 858.97 & 862.48 & 866.57 & 879.59 & 883.04 & 887.17 \\

\noalign{\vskip 2pt}
\thickhline
\end{tabular}

\endgroup
\end{table*}

\end{document}